\documentclass[11pt]{article}

\newif\ifdraft
\drafttrue %

\usepackage[english]{babel}
\usepackage[margin=1in]{geometry}
\usepackage{parskip}

\usepackage[colorlinks=true, allcolors=blue!50!black]{hyperref}
\usepackage{amsmath,amsthm, amssymb}
\usepackage{cleveref}
\usepackage{thm-restate}
\usepackage{nicefrac}

\newtheorem{theorem}{Theorem}[section]
\newtheorem{lemma}[theorem]{Lemma}

\newtheorem{proposition}[theorem]{Proposition}
\newtheorem{definition}[theorem]{Definition}

\newtheorem{corollary}[theorem]{Corollary}

\newtheorem{observation}[theorem]{Observation}

\usepackage[font={small,it}]{caption}

\usepackage{enumitem}

\usepackage{mathtools}
\renewcommand{\epsilon}{\varepsilon}

\usepackage[mathscr]{euscript}

\usepackage{graphicx}
\usepackage{xcolor}
\usepackage{caption}
\usepackage{subcaption}
\usepackage{tikz}
\usetikzlibrary{calc}
\usetikzlibrary{shapes.geometric}
\usetikzlibrary{arrows}
\usetikzlibrary{decorations.pathreplacing, decorations.pathmorphing}
\usetikzlibrary{patterns}
\usetikzlibrary{backgrounds}
\usetikzlibrary{scopes}
\usetikzlibrary{arrows, automata, positioning}
\usetikzlibrary{decorations.markings}

\numberwithin{equation}{section}

\allowdisplaybreaks

\newcommand{\RR}{\mathbb{R}}

\newcommand{\cA}{\mathscr{A}}

\newcommand{\cC}{\mathcal{C}}
\newcommand{\cD}{\mathcal{D}}

\newcommand{\cF}{\mathcal{F}}

\newcommand{\cI}{\mathcal{I}}

\newcommand{\cP}{\mathscr{P}}
\newcommand{\cQ}{\mathcal{Q}}

\newcommand{\cS}{\mathscr{S}}

\newcommand{\cU}{\mathcal{U}}

\renewcommand{\emptyset}{\varnothing}

\newcommand{\sse}{\subseteq}
\newcommand{\eps}{\varepsilon}

\ifdraft
\newcommand{\anupam}[1]{{\color{red}AG: #1}\marginpar{$\star\star$}}

\else
\newcommand{\anupam}[1]{}

\fi

\newcommand{\supp}{\operatorname{supp}}
\newcommand{\sep}{\operatorname{sep}}
\newcommand{\cost}{\operatorname{cost}}
\newcommand{\coverage}{\operatorname{coverage}}
\newcommand{\profit}{\operatorname{profit}}
\newcommand{\gain}{\operatorname{gain}}
\newcommand{\val}{\operatorname{value}}
\newcommand{\OPT}{\operatorname{OPT}}
\newcommand{\excess}{\operatorname{excess}}

\newcommand{\Usep}{\cU_{\rm{sep}}}
\newcommand{\Uunsep}{\cU_{\rm{unsep}}}

\newcommand{\ysep}{y(\Usep)}
\newcommand{\yunsep}{y(\Uunsep)}

\newcommand{\deac}{\mathbf{d}}
\newcommand{\nf}{\nicefrac}

\newcommand{\lineshere}{\medskip\hrule\vspace{0.01in}\hrule\medskip}

\newcommand{\appx}{\ensuremath{1.994}}

\title{Steiner Forest: A Simplified Better-Than-2 Approximation}
\author{
Anupam Gupta\thanks{
Computer Science Department, New York University, USA.
 Email: \href{mailto:anupam.g@nyu.edu}%
 {anupam.g@nyu.edu}. Supported in part by NSF awards CCF-2224718 and CCF-2422926.}
\and
Vera Traub\thanks{
Department of Computer Science, ETH Zurich, Switzerland.
Email: \href{mailto:vtraub@ethz.ch}%
{vtraub@ethz.ch}.
}
}
\date{}

\begin{document}

\maketitle

\begin{abstract}
  In the Steiner Forest problem, we are given a graph with edge
  lengths, and a collection of demand pairs; the goal is to find a
  subgraph of least total length such that each demand pair is connected in
  this subgraph. For over twenty
  years, the best approximation ratio known for the problem was a
  $2$-approximation due to Agrawal, Klein, and Ravi\ (STOC 1991), despite many
  attempts to surpass this bound. Finally, in a recent breakthrough,
  Ahmadi, Gholami, Hajiaghayi, Jabbarzade,
and Mahdavi\ (FOCS 2025) gave a $2-\eps$-approximation, where
  $\eps \approx 10^{-11}$. 
  
  In this work, we show how to simplify and extend the work of Ahmadi
  et al.\ to obtain an improved $1.994$-approximation. 
  We combine some ideas from their work (e.g., 
  an extended run of the moat-growing primal-dual algorithm, 
  and identifying autarkic pairs)
  with other ideas---submodular maximization to find components
  to contract, as in the relative greedy algorithms for Steiner
  \emph{tree}, and the use of autarkic \emph{triples}. We hope that our cleaner abstraction will open the way for further improvements.
\end{abstract}

\thispagestyle{empty}
\newpage

\setcounter{page}{1}

\section{Introduction}
\label{sec:introduction}

We study the Steiner Forest problem, where we are given an undirected
graph $G = (V,E)$ with non-negative edge costs $\{c_e\}_{e \in E}$,
and a collection of demand pairs $\cD \sse \binom{V}{2}$; a solution
is a subset of edges $F \sse E$ such that for each demand pair, both of its vertices lie in the
same connected component of the subgraph $(V,F)$. The goal
of the optimization algorithm is to minimize the total cost $c(F) \coloneqq \sum_{e
  \in F} c_e$.

The problem is APX-hard~\cite{ChlebikC08}, and the first constant-factor
approximation ratio given for this problem was the $2$-approximation
due to Agrawal, Klein, and Ravi~\cite{AKR95}, based on the (by-now standard)
moat-growing primal-dual framework. Remarkably, this approximation
ratio proved extremely difficult to improve upon: Goemans and
Williamson extended the primal-dual framework to a wide class of
constrained forest problems~\cite{GW95}, but the approximation ratio
for Steiner forest remained the same. (See \Cref{sec:related-work} for
a discussion of several other algorithms for the problem.)

Finally, in a remarkable result earlier this year, Ahmadi, Gholami,
Hajiaghayi, Jabbarzade, and Mahdavi~\cite{AGHJM25} gave a
$(2-\eps)$-approximation for $\eps \approx 10^{-11}$, breaking the
three-decades old logjam. Their work %
combines the primal-dual scheme with a local-search approach and a new
technique of ``autarkic pairs'', which builds on ideas from earlier
work of the same authors on the prize-collecting version of Steiner
Forest \cite{AhmadiGHJM25}.

 Building on this breakthrough, we show how to simplify and
  abstract these powerful new ideas and integrate them with classical
  approaches to get a more intuitive algorithm, a simpler analysis,
  and a stronger approximation guarantee.
Our main result is the
following:
\begin{theorem}
  There is an $\appx$-approximation algorithm for the Steiner Forest
  problem.
\end{theorem}

While building upon key ideas from~\cite{AGHJM25}, we achieve this
result by an arguably simpler algorithm and analysis. This
simplification, which unifies classical techniques with recent
breakthroughs, is a primary contribution of our work.  We show how to
adapt well-established tools from approximation algorithms for the
Steiner Tree problem (which is the special case of Steiner Forest
where all demand pairs share a common endpoint) to the setting of
Steiner Forest.  This leads not only to a significant simplification
and a better approximation ratio, but it also provides a more unified
perspective, revealing connections between the different algorithms we
have for these two problems.

\subsection{Our Techniques}

The first step of our new algorithm for Steiner Forest is to run the
$\epsilon$-extended moat growing algorithm introduced
in~\cite{AGHJM25}, which is a variant of the classical primal-dual
algorithm. In this variant sets stay active and grow for a bit longer
than they would in the classical primal-dual algorithm.  For a precise
description, we refer to \Cref{sec:eps-extended-moat}.

Next, we apply two different algorithms that each start with the solution resulting from the $\epsilon$-extended moat growing algorithm and aim at improving upon it.

\medskip\textbf{Improving Components:} The first of these two
algorithms aims at finding ``improving components'', which is a key
concept present in almost all known approximation algorithms for
Steiner Tree that have an approximation ratio below $2$, including
those algorithms achieving the currently best approximation guarantee
\cite{ByrkaGRS13,TraubZ25}.  Our algorithm builds on Zelikovsky's
relative greedy algorithm for Steiner Tree
\cite{zelikovsky_1996_better}, but with some key differences,
described next.

For the Steiner tree problem, there is a simple $2$-approximation
algorithm, which returns a spanning tree (in the metric closure) on
the set of \emph{terminals} (the set of vertices to be connected).
Zelikovsky's algorithm selects a set of ``components''---a set of
trees each connecting only a subset of the terminals---then contracts
each of the terminal sets connected by these components. It finally
builds a minimum spanning tree on the terminals in the contracted
instance to get a Steiner tree.

For the Steiner Forest problem, we do not have a simple
$2$-approximation via minimum spanning trees.  Thus, we instead complete the selected components via the primal-dual technique.  This
however creates a new challenge: In contrast to the cost of a minimum
spanning tree on the terminals, the cost of the Steiner forest
returned by the (classical or $\epsilon$-extended) primal-dual
algorithm might change in very non-intuitive ways when contracting
vertices. In fact, this cost is not even monotone and might
increase when contracting some vertices.  This makes it difficult to
analyze such an algorithm and also difficult to select a good set of
components.

For this reason, we work with a ``time-based'' version of the
primal-algorithm. To the best of our knowledge, such a timed version was first studied in~\cite{GuptaKPR03,KLS05}, and was rediscovered in
\cite{AGHJM25}.  Using this ``time-based'' version, the cost of a
solution returned by the primal-dual algorithm can only decrease when
contracting a component. Moreover, we show that the decrease can be
quantified by a submodular function. (See \Cref{def:gain} and
\Cref{lem:submod}.)  This allows us to apply techniques from
submodular optimization, much like the relative greedy algorithm used
by Zelikovsky \cite{zelikovsky_1996_better}, to select a good set of
components.

In order to be able to restrict ourselves to components connecting
only a constant-size set of vertices, which we need to obtain an
efficient algorithm, we use classical results by Borchers and Du on
$k$-restricted Steiner trees~\cite{BorchersD97}.  However, it turns
out that this restriction to constant-size sets is only possible when
our components are ``actively connected'' (a concept introduced in
\cite{AGHJM25}).  While this property is trivially satisfied for
Steiner Tree instances, this is not the case for Steiner Forest.  This
difficulty is the reason why we need the $\epsilon$-extended moat
growing algorithm, and cannot simply work with the classical
primal-dual algorithm.  We expand on this component-based algorithm in
\Cref{sec:activ-conn-comp}.

\medskip\textbf{Autarkic Pairs:} The second algorithm again starts
with the solution resulting from the $\epsilon$-extended moat-growing
algorithm.  It then aims at improvements of this solution via the idea
of ``autarkic pairs'' introduced in \cite{AGHJM25}, building on
similar ideas from \cite{AhmadiGHJM25}. The idea is to carefully
select a set of demand pairs to contract and then run any
$2$-approximation algorithm on the resulting instance. If the cost of
an optimum solution decreases by more than half of the cost needed to
connect the contracted demand pairs (by a shortest path each), then
this leads to a cheaper solution than applying a $2$-approximation
directly to the original instance.

We provide a new abstraction of the concept of autarkic pairs and show
how we can quantify the improvement in the solution cost for a given
collection of autarkic pairs (\Cref{lem:autarkic_pairs}).  Then we
prove that an optimal collection of autarkic pairs can be computed in
polynomial time (\Cref{thm:aut-algo}).  Moreover, our new abstraction
allows us to generalize our results to autarkic triples, which leads
to further improvements of the approximation guarantee.  For details
we refer to \Cref{sec:autarkic_pairs}.

In the end, we return the best of all the solutions we obtained from
these different algorithms.  For a description of our overall
algorithm and the key technical statement that implies the
approximation ratio of $\appx$, we refer to \Cref{sec:main-algo}.

\medskip\textbf{Plan for the Paper:} We discuss the $\eps$-extended
moat-growing algorithm in \Cref{sec:eps-extended-moat}, and explore
the ideas of improving components in \Cref{sec:activ-conn-comp}, and
autarkic pairs in \Cref{sec:autarkic_pairs}. We then give the
algorithm in \Cref{sec:main-algo}. The proof that one of these two
ideas gives an improvement appears
in \Cref{sec:satisfied_pairs_or_large_excess}, which we use to
complete the analysis in \Cref{sec:approximation-ratio}. We close with
some deferred proofs in the Appendices.

\subsection{Related Work}
\label{sec:related-work}

Apart from the recent work of Ahmadi, Gholami, Hajiaghayi, Jabbarzade,
and Mahdavi~\cite{AGHJM25} mentioned above, there is a large body of
work on Steiner network problems. The first approximation algorithms
for the Steiner forest problem was due to Agrawal, Klein, and
Ravi~\cite{AKR95}. The primal-dual approach they developed was
extended to a wide class of network design problems by Goemans and
Williamson~\cite{GW95}, who gave a network-formation plus
reverse-delete perspective for Steiner forest (as opposed to the
\cite{AKR95} algorithm which built edges irrevocably, and did not have
a reverse-delete step).  Another $2$-approximation algorithm for the
Steiner Forest problem is Jain's seminal iterative rounding
algorithm~\cite{Jain98}, which applies far beyond Steiner Forest to a
wide class of network design problems.

``Timed'' versions of the primal-dual algorithm had been proposed by
\cite{GuptaKPR03} in order to give an algorithm for a rent-or-buy
variant of the Steiner forest problem. K\"onemann, Leonardi, and
Sch\"afer proposed a timed algorithm to give another $2$-approximation
algorithm for Steiner forest which also admits a cross-monotonic
cost-sharing scheme~\cite{KLS05}. Their algorithm gave a new LP
relaxation, but this was shown to have an integrality gap no better
than a factor of two~\cite{KLSZ08}. Such a timed version of the
algorithm was rediscovered by Ahmadi et al.~\cite{AGHJM25}; one
component of their improvement was to perform a local search on the
time vectors. Other approaches for Steiner forest were explored,
including a greedy ``gluttonous'' algorithm~\cite{GK15-stoc} and a
local-search algorithm~\cite{GGKMSSV18}. While these approaches gave
constant-factor approximations, none of them were successful in
breaking past the factor of $2$.

In contrast, the Steiner \emph{Tree} problem, where all the demand
pairs share a common endpoint, has a simple $2$-approximation by just
outputting an MST on the terminals. Much better approximation
algorithms are known: a line of work culminating in the work of Robins
and Zelikovsky~\cite{RobinsZ05} gave \emph{relative greedy} algorithms
(based on finding small sets of terminals to contract, thereby
reducing the cost of the MST on the remaining instance); the last of
these algorithms gave an approximation ratio of $1.55$. 
Byrka, Grandoni,
Rothvo\ss, and Sanit\'a~\cite{ByrkaGRS13} gave a randomized version of
the relative greedy algorithm achieving an approximation of
$\ln 4 + \eps$, for any constant $\eps > 0$. (Goemans, Olver,
Rothvo\ss, and Zenklusen~\cite{GoemansORZ12} gave an LP-based analysis
of the same algorithm.) Recently, Traub and Zenklusen~\cite{TraubZ25}
gave a non-oblivious local-search algorithm with the same
approximation bound of $\ln 4 + \eps$.

The autarkic pairs technique from \cite{AGHJM25}, which we generalize
in \Cref{sec:autarkic_pairs}, builds on an idea introduced by the same
authors in their exploration of the prize-collecting generalization of
Steiner Forest (where we pay a penalty for every demand pair not
connected in the solution); they give the currently best-known
approximation ratio of $2$ for the problem \cite{AhmadiGHJM25}.  This
algorithm recursively applies the classical the primal-dual algorithm,
which has been extended to the prize-collecting version of Steiner
Forest in \cite{HajiaghayiJ06}.  A key step of the $2$-approximation
algorithm for prize-collecting Steiner Forest is the contraction of a
carefully selected set of demand pairs. The set of demand pairs to
contract is chosen such that the cost of an optimum solution decreases
a lot, unless the optimum solution is so expensive that the solution
returned by the primal-dual algorithm is already a $2$-approximation.
Similar techniques have been applied to the prize-collecting Steiner
tree problem, leading to an approximation factor of $1.79$
\cite{AhmadiGHJM24}.

The standard linear programming relaxation for Steiner Forest (see
\Cref{sec:notation}), which is used by the primal-dual algorithm, has
an integrality gap of exactly $2$ even for the MST/Steiner Tree
instances on the cycle.  A stronger linear programming relaxation
based on the Bidirected Cut Relaxation for Steiner Tree has been
investigated in \cite{ByrkaGT25}, but it is not known if this
relaxation has an integrality gap strictly smaller than $2$.  For the
special case of Steiner tree, Byrka, Grandoni and Traub recently
showed that the Bidirected Cut relaxation has an integrality gap below
$2$ \cite{Byrka0T24}.  This proof is also based on an extension of the
primal-dual algorithm where sets grow longer than in the classical
primal-dual algorithm, but it is nevertheless very different from the
$\epsilon$-extended moat-growing algorithm we use.  For Steiner Tree,
also other relaxations have been studied and it has been shown that
Hypergraphic Relaxation has an integrality gap of at most $\ln 4$
\cite{GoemansORZ12} (see also \cite{TraubZ25}).

\subsection{Notation and Basic Definitions}
\label{sec:notation}

In the \emph{Steiner Forest problem}, we are given a graph $G = (V,E)$
with non-negative edge costs $\{c_e\}_{e \in E}$, and a collection
$\cD = \{a_i,b_i\}_{i=1}^h$ of demand pairs, where each
$a_i, b_i \in V$. 
We say that $a_i$ and $b_i$ are partners for each
other. 
A vertex might have several partners.
The vertices belonging to pairs in $\cD$ are called
terminals. (We can assume without loss of generality that all vertices
in $V$ are terminals; indeed, any Steiner nodes can be replaced by a
demand pair at distance zero from each other.) 
A solution to the Steiner Forest problem is a collection $F\sse E$ of edges such that for each demand pair, both of its vertices lie
in the same connected component of the edge-induced graph $(V,F)$;
the goal is to find a solution of minimum total cost
$c(F) \coloneqq \sum_{e \in F} c_e$.

We use $\OPT$ to denote a fixed optimum solution.
If there are several optimum solutions, we choose an inclusionwise minimal one, 
that is an optimum solution that contains no edges of cost $0$ that could also be omitted.

A set $U \sse V$ \emph{separates} a demand $d = \{a,b\} \in \cD$ if
$d \cap U \neq \emptyset$ and $d \setminus U \neq \emptyset$. We
generalize this from demand pairs $d$ to a general set $S$.
\begin{definition}[Separation]
  A set $U \sse V$ \emph{separates} a set $S \sse V$ if
  $S\cap U\neq \emptyset$ and $S\setminus U\neq \emptyset$.
\end{definition}
Note that this definition is asymmetric: e.g., if $S = U \cup \{x\}$
for some $x \not\in U$, then $S$ is separated by $U$ (since there are
parts of $S$ both inside and outside $U$), but $U$ is not separated by
$S$. Define a \emph{separating cut} as a set $U \sse V$ such that
there exists some demand $d \in \cD$ separated by $U$.

The standard linear programming relaxation for Steiner forest has a
variable $x_e$ for each edge $e \in E$, with the following constraints:
\begin{align*}
  \min \sum_{e \in E} c_e & x_e \\
  \sum_{e \in \delta(U)} x_e &\geq 1 \qquad\qquad \forall \text{ separating
                              cuts } U \\
  x_e &\geq 0.
\end{align*}
Its LP dual has variables $y_U$ for each subset $U \sse V$ of
vertices:
\begin{align*}
  \max \sum_{\text{separating cuts } U} y_U & \\
  \sum_{U: e \in \delta(U)} y_U &\leq c_e \qquad \qquad \forall e \in E\\
  y_U &\geq 0.
\end{align*}
Note that we allow variables for all subsets $U$, while the objective
function only sums over the contribution for separating cuts; hence an
optimal solution has no incentive to have non-zero values for $y_U$'s
with non-separating sets $U$. However, the dual solutions constructed by
the algorithm in \Cref{sec:eps-extended-moat} have non-separating cuts in
their support.

\section{The $\eps$-Extended Moat-Growing Algorithm}
\label{sec:eps-extended-moat}

The following variant of the standard moat-growing algorithm of
\cite{AKR95, GW95} was considered by \cite{AGHJM25}. At each time $t$,
we have a collection of tight edges $F \sse E$, and let $\cC^t$ be the
connected components of the graph $(V,F)$. Moreover, we have
non-negative variables $\{y_S(t)\}_{S \sse V}$, such that tight edges
satisfy $\sum_{S \subseteq V: e \in \delta(S)} y_S(t) = c_e$ and all other edges satisfy $\sum_{S \subseteq V: e \in \delta(S)} y_S(t) < c_e$.

Each component of $\cC^t$  has a budget $b_S(t) \geq 0$. Initially
$y_S(t) = 0$ for all $S$ and $F = \emptyset$, which means the
components of $\cC^0$ are the individual vertices of the
graph. The budget of each terminal is initialized to zero.

A component $S \in \cC^t$ is called \emph{demand-active} if $S$
contains a terminal that has a partner to which it is not connected in $(V,F)$, and $S$ is called
\emph{budget-active} if it is not demand-active but it has a strictly
positive budget. Component $S \in \cC^t$ is called \emph{active} if it
is either demand-active or budget-active, and is called
\emph{inactive} if it is neither. Let $\cA^t \sse \cC^t$ be the
active components at time $t$.

Now we raise time $t$ continuously at unit rate and raise the dual
variables $y_S(t)$ of all active components $S \in \cA^t$ at the same
rate; the $y_S(t)$ values of all other sets $S$ remain unchanged. (In
other words, the derivative satisfies $\dot{y}_S(t) = \mathbf{1}(S \in \cA^t)$.) If the
component is demand-active, we increase the budget $b_S(t)$ at rate
$\eps$, else if it is budget-active, we decrease the budget $b_S(t)$
at unit rate. This process continues until some edge $e$ becomes
tight: i.e., $\sum_{S \in \cA^t: e \in \delta(S)} y_S(t) = c_e$. At
this moment, all newly-tight edges are added to $F$, which changes the
component structure. If components $S_1, \ldots, S_k$ are merged into
one component $S$, we merge their budgets, and hence define
$b_{S}(t) := \sum_{j = 1}^k b_{S_j}(t)$. This process can continue
indefinitely, but observe that the variables $y_S(t)$ stop growing
when there are no more active components.

Henceforth, we refer to this algorithm as the \emph{$\eps$-extended
  moat-growing algorithm}, or simply the \emph{extended algorithm} for
brevity. To avoid superfluous notation, we let $\{y_S\}_{S \sse V}$
denote the solution it generates when it terminates; i.e.,
$y_S = y_S(\infty)$. Note that
\begin{gather}
  \sum_{S\subseteq V} y_S = \int_{t=0}^\infty |\cA^t| \, dt.
\end{gather}

\subsection{Properties of the $\eps$-Extended Algorithm}
\label{sec:prop-eps-extend}

Since we grow duals based on either the demands or the budgets, we
partition the support of $y$ (denoted by $\cU$ or $\supp(y)$) into
\begin{align*}
  \Usep \ &\coloneqq\ \{ U \subseteq V : U \in
                            \supp(y) \text{ and separates some
                            demand pair}\} \\
  \mathcal{U}_{\rm unsep} \ &\coloneqq\ \{ U \subseteq V : U \in
                               \supp(y) \text{ and separates no demand pair}\}.
\end{align*}
Note that the sets in $\Usep$ contribute to the dual objective
function, whereas the others do not. By construction, the sets in the
support have a laminar structure.

Letting $y(\cU') = \sum_{U\in \cU'} y_U$ for any $\cU' \sse \cU$, we have 
\begin{gather}
  \yunsep = \eps \cdot \ysep = \frac{\eps}{1+\eps} \cdot y(\cU).
\end{gather}
Looking at which duals pack into $\OPT$, we get
\begin{equation}
  \label{eq:trivial_lower_bound}
  c(\OPT) \ \geq\ \sum_{U \sse V} y_U \,|\OPT \cap\ \delta U|
  \geq\ \ysep + \underbrace{\sum_{\substack{U\subseteq
        V:\\ |\OPT\cap \delta(U)|\geq 2}} y_U}_{=: \Lambda}.
\end{equation}

\begin{definition}
    Let $F\subseteq E$ be a solution to our Steiner forest instance. 
    We define its \emph{excess} (with respect to the dual solution
    $y$) to be
    \[
    \excess(F)\ \coloneqq\ c(F) - \ysep. \label{eq:excess}
    \]
\end{definition}
Now (\ref{eq:trivial_lower_bound}) says that $\excess(\OPT) \geq \Lambda$.
Moreover, the algorithm of \cite{AKR95,GW95} gives a
solution to Steiner forest with the following guarantee: 
\begin{theorem}
  \label{thm:AKR-GW}
  There is an efficient algorithm to find a Steiner forest of cost at most
  \[ 2\cdot \sum_{S \sse V} y_S \ \leq\ 2(1+\eps)\cdot \ysep \leq 2(1+\eps) \cdot \big( c(\OPT) - \excess(\OPT) \big). \]
\end{theorem}
If the optimal solution has a large excess, then the solution of the
$\eps$-extended algorithm itself is a good approximation. If not, we
give ways to improve it in the following sections.

\section{Actively Connected Components}
\label{sec:activ-conn-comp}

For a vertex $v \in V$, we define its \emph{deactivation time}
$\deac_v$ in the extended algorithm to be the largest time $t$ such
that the vertex $v$ belongs to an active component at time $s$ for all
$s < t$.  In other words, $v \in S$ for some $S \in \cA^s$ for all
times $s < t$, but not at time $t$ itself.  Observe that for a vertex
$v$, the deactivation time is \emph{at least} the first time at which
it lies in the same component as all its partners.  (The component
containing $v$ may remain active when $v$ is connected to its
partner---either demand-active, or budget-active if $\eps > 0$---and
hence typically the deactivation time will be larger.)

\begin{definition}[Actively Connected Vertices]
\label{def:AC-vtxs}
  Vertices $u,v \in V$ are \emph{actively connected} if there is a
  time $t$ during the run of the extended algorithm such that 
  \begin{enumerate}[nosep]
  \item  $\{u,v\} \sse S$ for some $S \in \cC^t$, and
  \item the deactivation time for both $u$ and $v$ is at least $t$.
  \end{enumerate}
\end{definition}
In other words, actively connected vertices belong to the same
connected component (which may or may not be active) at some time $t$,
and both elements have always been part of some active component
before time $t$. By definition, the vertices $a_i, b_i$ of each demand pair are always
actively connected. A key property of actively connected vertices is
that they have the same deactivation time. 

\begin{observation}\label{obs:same_deactivation_time}
If $u,v$ are actively connected, then $\deac_u =\deac_v$. 
\end{observation}
\begin{proof}
  If $u$ and $v$ are actively connected, they belong to the same
  connected component $S$ at some time
  $t \leq \min(\deac_u, \deac_v)$.  Either this component is inactive,
  in which case $\deac_u = \deac_v = t$, else they will always remain
  part of the same connected component henceforth.  Thus, their
  deactivation time is the same.
\end{proof}

Being actively connected is an equivalence relation: indeed, it is a
symmetric relation; moreover, if $\{u,v\}$ and $\{v,w\}$ are actively
connected, all three of $\{u,v,w\}$ have the same deactivation time,
and lie in the same connected component at this time. We can now
extend \Cref{def:AC-vtxs} to sets:

\begin{definition}[Actively Connected Sets]
  \label{def:activ-conn-sets}
  A \emph{set} $A \sse V$ is \emph{actively connected} if $A$
  is a subset of some equivalence class of the relation.
\end{definition}

The maximal actively connected sets form a partition of
the vertex set $V$.

\subsection{Contracting Actively Connected Components}
\label{sec:contraction_instead_of_boosting}

The next idea is a natural
one: if we can find an actively connected set $S$, which intersects
many moats during the moat-growing process, and shrink it down to a
single point, we can reduce the cost of the resulting solution. (Of course, we
need to pay the cost to shrink the set!) The idea extends to
collections of actively connected sets---and since the ``gain'' we
achieve is submodular, we can use the machinery of submodular
maximization to find a good collection of such sets to contract. 

\begin{definition}[Gain]\label{def:gain}
  Let $\cS$ be a collection of vertex sets, each being
  actively connected. Define  
  \begin{gather}
    \gain(\cS)\ \coloneqq \ 2 \cdot \int_{t=0}^{\infty} \big( |\cA^t| -|\cA^t / \cS| \big)\, dt,
  \end{gather}
  where $\cA^t / \cS$ arises from $\cA^t$ by
  merging sets in $\cA^t$ into a single set whenever they are
  intersected by the same set $S\in \cS$.
\end{definition}

\begin{lemma}[Submodularity]
  \label{lem:submod}
  The function
  $\gain \colon 2^{\{ S \subseteq V \colon S\text{ actively
      connected}\}} \to \mathbb{R}_{\geq 0}$ is submodular.
\end{lemma}
\begin{proof}
  To prove that the $\gain$ function is submodular, we show that for
  each time $t\geq 0$ the function $g^t$ given by
  $g(\cS) \coloneqq |\cA^t| -|\cA^t / \cS| $ is submodular.  This
  implies that the function $\gain$ given by
  $\gain(\cS) = \int_{t=0}^{\infty} g^t(\cS) \ dt$ is submodular.  To
  prove that the function $g^t$ is submodular it suffices to prove
  that it satisfies the diminishing return property: for any two
  collections $\cS_1$ and $\cS_2$ of actively connected sets with
  $\cS_1 \subseteq \cS_2$ and any actively connected set $S$, we have
  \begin{equation}\label{eq:diminishing_returns}
    g^t(\cS_1 \cup \{S\}) - g^t (\cS_1) \geq g^t(\cS_2 \cup \{S\}) - g^t(\cS_2).
  \end{equation}
  For $i\in \{1,2\}$, let $m_i$ be the number of sets from
  $\cA^t / \cS_i$ intersected by $S$.  Then the marginal increase
  $g^t(\cS_i \cup \{S\}) - g^t (\cS_i)$ is $m_i -1$ unless $m_i =0$,
  in which case the marginal increase is $0$.  Because
  $\cS_1 \subseteq \cS_2$, every set in $\cA^t / \cS_2$ is a union of
  sets in $\cA^t / \cS_1$, implying $m_2 \leq m_1$.  This shows
  \eqref{eq:diminishing_returns}.
\end{proof}

\begin{lemma}
  \label{lem:gain-gain}
  Let $\cS$ be a collection of vertex sets, each being actively
  connected.  Let $\cI/\cS$ be the instance arising from the
  original instance $\cI$ by contracting each set $S\in \cS$.
  Then we can efficiently compute a solution for the instance
  $\cI/\cS$ of cost at most
  \begin{gather}
    2\cdot y(\cU) - \gain(\cS),  \label{eq:gain-gain}
  \end{gather}
  where $y$ is solution from the extended algorithm.
\end{lemma}

\begin{proof}
  Consider a time-based moat-growing algorithm that receives the
  vector $\deac \in \mathbb{R}^V$ of deactivation times of the vertices
  as input, and increases a variable $y_S$ corresponding to a
  connected component whenever the current time $t$ satisfies
  $t < \deac_v$ for at least one vertex $v \in S$.  More precisely,
  we again start with the vector $y(t)=0$ and let $\cC^t$ be the
  vertex sets of the connected components of the subgraph $(V,F)$,
  where $F$ is the set of tight edges at time $t$.  A component
  $S\in \cC^t$ is active if $t< \deac_v$ for at least one vertex
  $v \in S$.  Now we raise time $t$ continuously at unit rate and
  raise the dual variables $y_S(t)$ of all active components
  $S \in \cA^t$ at the same rate; the $y_S(t)$ values of all other
  sets $S$ remains unchanged.  Once $t \geq \max_{v \in V} \deac_v$,
  the vector $y(t)$ does not change anymore and this is what the
  algorithm returns.  This time-based moat-growing algorithm produces
  exactly the same vector $\{y_S\}_{S \sse V}$ as the extended
  algorithm.

  By \Cref{obs:same_deactivation_time}, all vertices contained in the
  same set $S\in \cS$ have the same deactivation time.  Thus all
  vertices $v$ in the instance $\cI$ that are contracted into the same
  vertex of the instance $\cI/\cS$ have the same deactivation time.
  We define the deactivation time $\deac'_w$ for a vertex $w$ in the
  contracted instance as the deactivation time $\deac_v$ of the
  original vertices $v$ contracted into $w$.  We now run the
  time-based moat-growing algorithm on the instance.
    
  Let $y(t)$ and $y'(t)$ denote the dual variables at time $t$ in this
  time-based run on the original instance $\cI$ and the contracted
  instance $\cI/\cS$, respectively, and $\cA^t$ denote the active sets
  at time $t$ in the time-based run on $\cI$.  We show the following
  invariants for any time $t\geq0$:
  \begin{enumerate}[ label=(\roman*), nosep]
  \item \label{item:same_active} The active sets in the time-based run
    on $\cI/\cS$ at time $t$ are exactly the sets in $\cA^t/\cS$.
  \item \label{item:tight_same}
   For every edge $e$ that does not have both endpoints in the same set from $\cA^t / \cS$, we have
  \[
    \sum_{S \subseteq V: e \in \delta(S)} y'_S(t)\  = \sum_{S \subseteq V: e \in \delta(S)} y_S(t).
  \]
  \end{enumerate}  
  These invariants are trivially satisfied for $t=0$. Note that the
  active sets only change a finite number of times.  If the invariants
  hold at some time $t$, then growing the dual variables corresponding
  to the active sets until some new edge becomes tight (in one of the
  runs of the algorithms) maintains property~\ref{item:tight_same},
  because of property~\ref{item:same_active}.  Then the active sets
  might change; however, by property~\ref{item:tight_same}, every edge
  that does not have both endpoints in the same set from $\cA^t / \cS$
  is tight in one of the two runs of the algorithm if and only if it
  is tight in the other.  Finally, \Cref{obs:same_deactivation_time}
  implies that all vertices contained in the same set $S\in \cS$ have
  the same deactivation time, hence \ref{item:same_active}
  is maintained.

Invariant \ref{item:same_active} implies that every vertex that is not yet connected
to all of its partners at some time $t$ is part of an active set at time $t$ when running
the time-based algorithm on the contracted instance $\cI/\cS$.
Thus, the primal-dual algorithm from  \cite{AKR95,GW95} yields a feasible solution $H$
of cost at most
\[
c(H)\ \leq \ 2 \cdot \int_{t=0}^{\infty} |\cA^t / \cS| \ dt \ =\ 2
\cdot \int_{t=0}^{\infty} |\cA^t| \ dt - \gain(\cS) \ =\  2\cdot
y(\cU) -\gain(\cS). \qedhere
\]
\end{proof}

\begin{definition}[Cost of a vertex set]
  For a set $S\subseteq V$, let $\cost(S)$ be the cost of a cheapest
  Steiner tree for $S$, that is the minimum cost of an edge set
  $F\subseteq E$ such that all vertices of $S$ belong to the same
  connected component of $(V,F)$.  For a collection $\cS$ of vertex
  sets, we define $\cost(\cS) \coloneqq \sum_{S\in \cS} \cost(S)$.
\end{definition}

\subsection{An Efficient Algorithm}
\label{sec:effic-AC}
  
These results suggest the following algorithm: find a collection $\cS$
of actively connected sets for which $\gain(\cS)- \cost(\cS)$ is
large, build Steiner trees on each of them, and then use
\Cref{lem:gain-gain} to find a solution on the contracted instance.
In order to make the search over actively connected sets tractable, we
can use known results on $k$-restricted Steiner trees to restrict
ourselves to collections $\cS$ of actively connected sets that contain
only a constant number of vertices, as follows:

\begin{lemma}\label{lem:restriction_to_constant_size}
  Let $k \in \mathbb{N}_{\geq 2}$.  For every actively connected set
  $S$, then there exists a collection $\cS_k$ of subsets of $S$, each
  containing at most $k$ elements, such that
  \begin{itemize}[nosep]
  \item $\cost(\cS_k) \leq \big(1+\frac{1}{\lfloor \log_2 k \rfloor}\big) \cdot \cost(S)$, and
  \item ${\rm gain}(\cS_k) = {\rm gain}(\{S\})$.
  \end{itemize}
\end{lemma}

\begin{proof}
  Let $T$ be a cheapest Steiner tree for the set $S$. Then
  $\cost(S)=c(T)$. By a result due to Borchers and
  Du~\cite{BorchersD97}, there exists a $k$-restricted Steiner tree
  $T_k$ for $S$ satisfying
  $c(T_k) \leq \big(1+\frac{1}{\lfloor \log_2 k \rfloor}\big) \cdot
  \cost(S)$.  In other words, there exists a collection $\cS_k$ of
  subsets of $S$, each containing at most $k$ elements, such that
  \begin{enumerate}[nosep,label=(\roman*)]
  \item\label{item:hypergraph_connectivity} the hypergraph with vertex set $S$ and hyperedges being the sets in $\cS_k$ is connected, and
  \item $\cost(\cS_k) \leq \big(1+\frac{1}{\lfloor \log_2 k \rfloor}\big) \cdot \cost(S)$.
  \end{enumerate}
  In order to prove $\gain(\cS_k) =\gain(\{S\})$, it suffices to show
  that for every time $t\geq 0$, we have
  $|\cA^t| -|\cA^t / \cS_k| \ =\ |\cA^t| -|\cA^t / \{S\}|$.  We recall
  that $\cA^t / \{S\}$ arises from $\cA^t$ by merging sets in $\cA^t$
  into a single set whenever they are intersected by the same set $S$.
 
  By the definition of $S$ being actively connected, there exists a
  time $t^*$ such that $S\subseteq C$ for some $C\in \cC^{t^*}$, and
  moreover, all vertices in $S$ must be in active components at each
  time $t< t^*$. The former fact means that $S$ is contained in the
  same connected component from $\cC^t$ at all times $t \geq t^*$,
  implying $|\cA^t| -|\cA^t / \cS_k| = 0 = |\cA^t| -|\cA^t / \{S\}|$.
  For times $t < t^*$, each vertex in $S$ is contained in some set
  from $\cA^t$.  Thus, by \ref{item:hypergraph_connectivity}, we have
  $\cA^t / \cS_k = \cA^t/\{S\}$, completing the proof. (See 
  \Cref{sec:crit-active-conn} for an example where this last
  equality fails if $S$ is not actively connected.)
\end{proof}

This means that modulo an arbitrarily small loss, it suffices to
consider actively connected sets connecting at most $k$ actively
connected vertices. Now, to approximately maximize the ``profit''
(i.e., $\gain(\cS) -\cost(\cS)$), we can use the submodularity of the
$\gain$ function to obtain the following theorem, which allows us to
improve upon the cost $2\cdot y(\cU)$ of the solution obtained from
the extended moat-growing algorithm whenever there exists a collection
$\cS$ of actively connected sets with $\gain(\cS) > \cost(\cS)$.

\begin{theorem}
  \label{thm:effic-contraction}
  Let $\cI$ be a given of the Steiner Forest problem and $\cS$ an (unknown) collection of vertex sets actively connected by
  the extended algorithm. 
  Then, for any $\alpha \geq 0$ and any $\delta > 0$, we can find a Steiner forest
  for $\cI$, having cost at most
  \begin{gather*}
    2 \cdot y(\cU) + (\alpha+\delta) \cdot \cost(\cS) - (1-e^{-\alpha})\cdot \gain(\cS)
  \end{gather*}
  in polynomial time.
\end{theorem}
\begin{proof}
  Let $k \in \mathbb{N}$ be large enough such that
  $(\alpha + \nf{\delta}2) \cdot \big(1+\frac{1}{\lfloor \log_2 k \rfloor}\big) \leq
  \alpha+ \delta$.  We say that a collection of actively
  connected vertex sets is $k$-restricted if each of its sets has
  cardinality at most $k$, i.e., it contains only sets from
  \[
    \cA_k \coloneqq \big\{ S\subseteq V \colon S\text{ actively connected and }|S|\leq k \big\}.
  \]
  By \Cref{lem:restriction_to_constant_size}, there exists a
  $k$-restricted collection $\cS_k \subseteq \cA_k$ of actively
  connected sets such that $\gain(\cS_k) =\gain(\cS)$ and
  $\cost(\cS_k) \leq \big(1+\frac{1}{\lfloor \log_2 k \rfloor}\big)
  \cdot \cost(\cS)$.  Because $|\cA_k| = O(n^k)$ and $k$ is a
  constant, this set has polynomial size.  Using that every set
  $A\in \cA_k$ has constant size, we can compute $\cost(A)$, that is
  the cost of a cheapest Steiner tree for $A$, for all $A\in \cA_k$ in
  polynomial time by \cite{DreyfusW71}.  The function
  $\cost : 2^{\cA_k} \to \mathbb{R}_{\geq 0}$ is modular (by
  definition) and by \Cref{lem:submod}, the function
  $\gain : 2^{\cA_k} \to \mathbb{R}_{\geq 0}$ is submodular. Thus,
  using ideas from submodular optimization subject to knapsack
  constraints (see \Cref{sec:submod}) we can compute a collection
  $\widetilde{\cS} \subseteq \cA_k$ satisfying
  \begin{gather*}
    \gain(\widetilde{\cS}) - \cost(\widetilde{\cS}) \ \geq\
    (1-e^{-\alpha})\cdot \gain(\cS_k) -  (\alpha + \nf\delta2) \cdot  \cost(\cS_k)
  \end{gather*}
  in time $O(n^{O(1/\delta)})$. By \Cref{lem:gain-gain}, we obtain a
  solution to the contracted instance $\cI/\cS$ of cost at most
  $2\cdot y(\cU) - {\rm gain}(\widetilde{\cS})$.  The union of this
  solution and cheapest Steiner trees for all sets
  $S\in \widetilde{\cS}$ is a feasible solution to the instance $\cI$
  of cost
  \begin{align*}
    2\cdot y(\cU) - {\rm gain}(\widetilde{\cS}) + \cost(\widetilde{\cS}) \
    \leq&\  2\cdot y(\cU) -  (1-e^{-\alpha})\cdot \gain(\cS_k) +
          (\alpha+\nf\delta2) \cdot  \cost(\cS_k)\\ 
    \leq&\  2\cdot y(\cU) + (\alpha+\delta)\cdot\cost(\cS) -
          (1-e^{-\alpha})\cdot \gain(\cS), 
 \end{align*}
 where we used $\gain(\cS_k) =\gain(\cS)$ and
 $(\alpha+\nf\delta2) \cdot \cost(\cS_k) \leq (\alpha+\nf\delta2)  \cdot \big(1+\frac{1}{\lfloor
   \log_2 k \rfloor}\big) \cdot \cost(\cS) \leq (\alpha +\delta) \cdot
 \cost(\cS)$.
\end{proof}

\subsection{The Path Forward}
\label{sec:path-forward}

To get a good approximation for Steiner forest, the natural next step
would be to show that \emph{there exists} some collection $\cS$ which
has a large gain, relative to the cost of connecting its sets.

This is easy for the special case of Steiner \emph{tree}: we can
choose $\cS$ as the single set containing all the terminals. This set
is actively connected, and has $\cost(\cS) \leq c(\OPT)$. The
corresponding gain is essentially all of $2\sum_{S\subseteq V} y_S$
(up to one set at a time, which is the slack we already have in the
primal-dual analysis). Optimizing the value of $\alpha$ in
\Cref{thm:effic-contraction} and choosing $\epsilon = 0$---i.e., the
classical primal-dual algorithm without $\epsilon$-extensions---yields
an approximation factor of $1+\ln 2 +\delta$, thus recovering the
result of Zelikovsky's relative greedy algorithm
\cite{zelikovsky_1996_better}.

In Steiner Forest, a natural strategy would be to choose $\cS$ as the
connected components of $\OPT$. However, there are two problems with
this.  Firstly, the vertex set of a connected component of $\OPT$ is
not necessarily actively connected. In fact, this is why we use the
$\epsilon$-extended moat-growing algorithm. Indeed, we would like sets
to stay active in the moat-growing process until they no longer
separate the vertex set of any component of $\OPT$. However, we do not
know $\OPT$'s components. The budget-active growth makes more vertices
actively connected.

Of course, this small additional growth may not suffice to actively
connect each component of $\OPT$. Nonetheless, the key observation is
this: for every set $U\in \Uunsep$, the number of edges of $\OPT$ in
the cut $\delta(U)$ is either $0$ or at least $2$. (Indeed, if
$\delta(U)$ had a single edge, it could be omitted and this contradicts the
minimality of $\OPT$.) In other words, if many grown budget-active
sets contain few edges of $\OPT$, this intuitively means that we have
managed to make a large fraction of the $\OPT$ components actively
connected. Otherwise, if the growth of budget-active sets is not
sufficient---and these cuts continue to contain many edges of $\OPT$
until the components becomes inactive---then we contribute to the
``excess'', and hence increase the lower bound on $c(\OPT)$ compared
to $y(\Usep)$.

\begin{figure}[t]
  \begin{center}
    \begin{tikzpicture}[scale=0.6]

\tikzset{vertex/.style={
fill=black,thick, circle,minimum size=5pt, inner sep=0pt, outer sep=1.5pt}
}

\begin{scope}[every node/.style={vertex}]
\node (a1) at (3,2) {};
\node (a2) at (5,2) {};
\node (a3) at (7,2) {};
\node (ak) at (10,2) {};

\node (b1) at (3,0) {};
\node (b2) at (5,0) {};
\node (b3) at (7,0) {};
\node (bk) at (10,0) {};

\node (s) at (1,0) {};
\node (t) at (12,0) {};
\end{scope}

\node[above=1pt] at (a1) {$a_1$};
\node[below=1pt] at (b1) {$b_1$};
\node[above=1pt] at (a2) {$a_2$};
\node[below=1pt] at (b2) {$b_2$};
\node[above=1pt] at (a3) {$a_3$};
\node[below=1pt] at (b3) {$b_3$};
\node[above=1pt] at (ak) {$a_k$};
\node[below=1pt] at (bk) {$b_k$};
\node[left=1pt] at (s) {$s$};
\node[right=1pt] at (t) {$t$};

\begin{scope}[very thick]
\begin{scope}[red!80!black]
\draw (a1) -- node[right] {1} (b1);
\draw (a2) -- node[right] {1}(b2);
\draw (a3) --node[right] {1} (b3);
\draw (ak) -- node[right] {1}(bk);
\end{scope}
\begin{scope}[blue!90!black]
\draw (s) -- node[below] {1} (b1);
\draw (b1) -- node[below] {1} (b2);
\draw (b2) -- node[below] {1} (b3);
\draw (b3) --  (8,0);
\draw (10,0) --  (bk);
\draw (bk) -- node[below] {1} (t);
\end{scope}
\draw[bend right, out=-80, in=-100, looseness=0.3, green!60!black] (s) to node[below] {2} (t);
\end{scope}

\node[blue] at (9,0) {$\dots$};

\begin{scope}[shift={(14,0)}]
\begin{scope}[every node/.style={vertex}]
\node (a1) at (3,2) {};
\node (a2) at (5,2) {};
\node (a3) at (7,2) {};
\node (ak) at (10,2) {};

\node (b1) at (3,0) {};
\node (b2) at (5,0) {};
\node (b3) at (7,0) {};
\node (bk) at (10,0) {};

\node (s) at (1,0) {};
\node (t) at (12,0) {};
\end{scope}

\node[above=1pt] at (a1) {$a_1$};
\node[below] at (b1) {$b_1$};
\node[above=1pt] at (a2) {$a_2$};
\node[below] at (b2) {$b_2$};
\node[above=1pt] at (a3) {$a_3$};
\node[below] at (b3) {$b_3$};
\node[above=1pt] at (ak) {$a_k$};
\node[below] at (bk) {$b_k$};
\node[left=1pt] at (s) {$s$};
\node[right=1pt] at (t) {$t$};

\begin{scope}[line width=2pt, orange]
\draw (a1) ellipse (0.9 and 0.9);
\draw (b1) ellipse (0.9 and 0.9);
\draw (a2) ellipse (0.9 and 0.9);
\draw (b2) ellipse (0.9 and 0.9);
\draw (a3) ellipse (0.9 and 0.9);
\draw (b3) ellipse (0.9 and 0.9);
\draw (ak) ellipse (0.9 and 0.9);
\draw (bk) ellipse (0.9 and 0.9);

\draw (s) ellipse (0.9 and 0.9);
\draw (t) ellipse (0.9 and 0.9);

\draw[dashed, cyan, opacity=0.6] (6.5,0.7) ellipse (7 and 3.6);
\end{scope}

\begin{scope}[orange]
\node at (3,3.3) {$A_1$};
\node at (5,3.3) {$A_2$};
\node at (7,3.3) {$A_3$};
\node at (10,3.3) {$A_k$};
\node at (3,-1.5) {$B_1$};
\node at (5,-1.5) {$B_2$};
\node at (7,-1.5) {$B_3$};
\node at (10,-1.5) {$B_k$};
\end{scope}

\end{scope}

\end{tikzpicture}
  \end{center}
  \caption{\label{fig:matching} An instance where the moat-growing
    algorithm actively connects all vertices, but no collection $\cS$
    of vertex sets exists with $\gain(\cS) > \cost(\cS)$.  The left
    figure shows the graph $G$ and the edge costs. The demand pairs
    are $\{ s,t \}, \{a_1, b_1\}, \dots, \{ a_k,b_k\}$. $\OPT$
    consists of the red matching and the green edge, and has cost
    $k+2$.  The primal-dual algorithm (classical or
    $\epsilon$-extended) returns a solution of cost $2k+1$, consisting
    of all red and blue edges. The right figure shows $\supp(y)$.  In
    the dual solution obtained from classical moat-growing, we have 
    $y_U=\frac{1}{2}$ for each orange set $U$; the $\epsilon$-extended
    algorithm also has a value of $\epsilon \cdot (k+1)$ for the dual
    variable corresponding to the light-blue set $V$.  }
\end{figure}
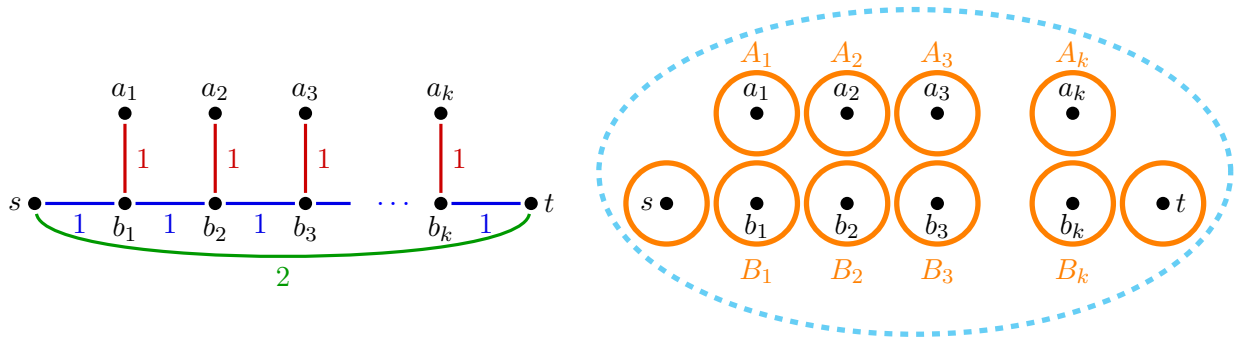

The second hurdle is that, even if the vertex set of every connected
component of $\OPT$ is actively connected, the gain we obtain from it
might be too small. More precisely, the gain might be only half of
$2\sum_{S\subseteq V} y_S$ if the components of $\OPT$ form a matching
on the elements of $\cA^t$.  See \Cref{fig:matching} for an example.
To gain from connected components in $\OPT$ that connect only two sets
from $\cA^t$, we use the second algorithm that profits from autarkic
collections, which we define next.

\section{Autarkic Collections}
\label{sec:autarkic_pairs}

We start off with defining the notion of autarkic pairs from \cite{AGHJM25},
which we then extend to triples to obtain our improved bounds. The
concept of an autarkic pair is motivated by the example in
\Cref{fig:matching}. To handle such instances, we would like to
identify a suitable subset of the demand pairs---in this case the
demand pairs $\{a_1, b_1\}, \dots, \{a_k, b_k\}$---contract them, and
apply a $2$-approximation algorithm to the resulting instance. Then
we obtain a solution to the original instance by adding a shortest
path between any two contracted vertices.

Now consider a variant of the same example containing multiple copies
of each demand pair $\{a_i,b_i\}$, where all copies of $a_i$ are at
distance $\delta \approx 0$ from each other, and the same for copies
of $b_i$. We cannot afford paying for a separate shortest path for
each copy, so we would want to contract a single one of these copies.
Thus, autarkic pairs are defined not as demand pairs, but as pairs of
sets (e.g., the pairs $\{A_i, B_i\}$ in the right figure); and we
later choose one demand pair as the designated representative for the
autarkic pair.  Only these designated representatives are contracted
to obtain the residual instance.

We now formalize this, and generalize it to triples. Given a set
$S \in \supp(y)$, define $\sep(S)$ to be the set of demand pairs
separated by $S$---i.e.,
\[ \sep(S) := \{ \{a_i, b_i\} \in \cD \mid |\{a_i,
b_i\} \cap S| = 1\}. \]

\begin{definition}[Autarkic Pairs]
  For sets $A, B \in \supp(y)$, define $P = \{A,B\}$ to be an
  \emph{autarkic pair} if $A \cap B = \emptyset$, both $A,B \in \cA^t$
  for some time $t$, and $\sep(A) = \sep(B) \neq \emptyset$.  We overload notation
  and define $\sep(P) := \sep(A) \cup \sep(B) = \sep(A)$.
\end{definition}

For our work, we generalize this notion to triples:
\begin{definition}[Autarkic Triples]
  Define $P = \{A_1,A_2,A_3\}$ to be an \emph{autarkic triple} if
  \begin{enumerate}[nosep]
  \item the sets $A_i$ are disjoint,
  \item there exists a time $t$ such that $A_i \in \cA^t$ for all $i$
    (and hence each $A_i \in \supp(y)$),
  \item $\sep(A_i) \neq \emptyset$, but
  \item $\sep(A_1 \cup A_2 \cup A_3) = \emptyset$.
  \end{enumerate}
  Define the separated pairs for $P$ to be
  $\sep(P) := \cup_i \sep(A_i)$.
\end{definition}
In words, the demands individually separated by these sets $A_i$ cross
between these sets, but there are no demands going outside their
union. We use the term \emph{autarkic tuples} to denote both pairs and
triples, and the term \emph{autarkic collection} for a collection of
such tuples.
  
\begin{definition}[Crossing-Free Collection]
  An autarkic collection $\cP$ is
  \emph{crossing-free} if %
  $\sep(P) \cap \sep(Q) = \emptyset$ for all $P\neq Q$ with $P,Q\in \mathcal{P}$.
\end{definition}

\begin{definition}[Coverage of $\cP$]
  We say that a set $U \in \supp(y)$ is \emph{covered by} a
  crossing-free autarkic collection $\cP$ if there exists a
  $P \in \cP$ and some set $A_i \in P$ such that
  $\sep(A_i) = \sep(U)$.  Given $y$, define the coverage of $\cP$ as
  \begin{gather}
    \coverage(\cP) := \sum_{U \sse V: \cP \text{ covers } U} y_U.
  \end{gather}
\end{definition}

\subsection{Connecting Autarkic Tuples and Residual Instances}
\label{sec:conn-autark}

For each autarkic tuple $P$, we define the \emph{connector} of $P$ as follows:
\begin{enumerate}
\item For each pair $P \in \cP$, pick a demand pair
  $\{a, b\} \in \sep(P)$ as its \emph{designated representative}.  
  Define the connector for $P$ as the direct edge $\{a,b\}$.
\item For each triple $P = (A_1, A_2, A_3) \in \cP$, for each
  $i \neq j \in \{1,2,3\}$ such that
  $\sep(A_i)\cap \sep(A_j) \neq \emptyset$, choose a demand pair
  $\{a_{ij}, b_{ij}\}$ from $\sep(A_i)\cap \sep(A_j)$ as the
  \emph{designated representative}. Define the connector for $P$ as
  the minimum Steiner tree on these (at least three and at most six)
  vertices which are designated representatives. %
\end{enumerate}

\begin{definition}[Cost and Profit for $\cP$]
  For an autarkic tuple $P$, let $\cost(P)$ be the total length of the
  connector for $P$; for a collection $\cP$, define
  $\cost(\cP) = \sum_{P \in \cP} \cost(P)$. Define
  \begin{gather}
    \profit(\cP) = 2\cdot\coverage(\cP) - \cost(\cP).
  \end{gather}
\end{definition}

\begin{restatable}[Efficient Algorithm]{theorem}{EffAut}
  \label{thm:aut-algo}
  There is an efficient algorithm to find a crossing-free autarkic collection
  $\cP$ that maximizes $\profit(\cP)$. %
\end{restatable}

The main idea of this algorithm is to show a ``laminar'' structure on
the autarkic tuples, which allows us to then find the best
crossing-free collection using dynamic programming; we defer the proof
to \Cref{sec:aut-algo}.  Given a crossing-free autarkic collection
$\cP$, we define the contracted instance $\cI/\cP$ as the instance
obtained by contracting the connectors for each $P \in \cP$, and the
cost $\cost(\cP) := \sum_{P \in \cP} \cost(P)$ as the sum of costs of
the connectors.  The following key lemma quantifies the gain
from~$\cP$:

\begin{restatable}{lemma}{UsingAut}
  \label{lem:autarkic_pairs}
  Let $\cP$ be a crossing-free autarkic collection.  Then
  there exists a solution for the contracted instance $\cI/\cP$ of
  cost at most
  \begin{gather}
    c(\OPT(\cI)) - {\rm coverage}(\cP) + \sum_{\substack{S\subseteq
        V:\\ |\OPT\cap \delta(S)|\geq 2}} y_S. \label{eq:1}
  \end{gather}
\end{restatable}

\begin{proof}
  Consider any set $U\subseteq V$ covered by $\cP$.  Then OPT contains
  at least one edge in $\delta(U)$.  Let $D$ be the set of edges
  $e \in \OPT$ such that for some set $U\subseteq V$ covered by $\cP$,
  we have $\OPT\cap \delta(U)=\{e\}$.  We claim that $\OPT\setminus D$
  is a feasible solution to the contracted instance $\cI / \cP$.
  Because $y$ is a feasible dual LP solution, we have
  \[
    c(D) \ \geq\ {\rm coverage}(\cP) - \sum_{\substack{U\subseteq V:\\
        |\OPT\cap \delta(U)|\geq 2}} y_U. 
  \]
  Thus, it remains to prove that $\OPT\setminus D$ is indeed feasible
  for the contracted instance $\cI / \cP$.

  We consider some demand pair $\{v,w\}$ and show that
  $\OPT\setminus D$ contains a $v$-$w$ path in the contracted
  instance.  Because $\OPT$ is a feasible solution for the instance
  $\cI$, it contains some $v$-$w$ path $R$.  If there is some edge
  $e\in R\cap D$, then there must be some set $U\in \supp(y)$ that is
  covered by $\cP$ and satisfies $\OPT \cap \delta(U) =\{e\}$.  The
  $v$-$w$ path $R$ in $\OPT$ contains exactly one edge in $\delta(U)$,
  implying $\{v,w\}\in \sep(U)$.  Therefore, $U$ must be covered by
  some autarkic tuple $P^*\in \cP$ with $\{v,w\}\in \sep(P^*)$; this
  tuple is unique, due to $\cP$ being crossing-free. Moreover, by
  definition of autarkic tuples, we have sets $A_1,A_2\in P^*$ with
  $\{v,w\}\in \sep(A_1)\cap \sep(A_2)$, where without loss of
  generality $v\in A_1$ and $w\in A_2$.  Consider the collection $\cC$
  of sets $U\in \supp(y)$ with $\sep(U)=\sep(A_1)$ or
  $\sep(U)=\sep(A_2)$.  We write $D^*\subseteq D$ to denote the set of
  edges $e \in \OPT$ such that for some set $U$ in this subcollection
  $\cC$, we have $\OPT\cap \delta(U)=\{e\}$. Observe that
  $R\cap (D\setminus D^*)=\emptyset$.

  Recall from \Cref{sec:conn-autark} that the instance $\cI/\cP$ is
  based on contracting the connectors built on the designated
  representatives. Let $\{v^*,w^*\}\in \sep(A_1)\cap \sep(A_2)$ be the
  designated representative with $v^*\in A_1$ and $w^*\in A_2$.  Then
  the $v^*$-$w^*$ path $R^*$ satisfies
  $R^*\cap (D\setminus D^*)=\emptyset$ by the same argument as for the
  $v$-$w$ path $R$.  We claim that $(R\cup R^*)\setminus D^*$ contains
  a $v$-$v^*$ path---and by symmetry, $(R\cup R^*)\setminus D^*$ also
  contains a $w$-$w^*$ path. This claim will complete the proof
  because $v^*$ and $w^*$ are contracted in $\cI/\cP$.

  To prove the claim, observe that because $\supp(y)$ is laminar, the
  cuts induced by the collection $\cC$ correspond to a chain (up to
  taking complements of sets).  Recall that $A_1,A_2\in \cC$ with
  $v,v^*\in A_1$ and $w,w^*\in A_2$ and moreover, both $\{v,w\}$ and
  $\{v^*,w^*\}$ are separated by each set $U\in \cC$.  Thus, every set
  $U\in \cC$ separates $v$ and $v^*$ from $w$ and $w^*$.  Consider the
  cut $C$ induced by $\mathcal{C}$ with $|\OPT\cap C|=1$ that is
  closest to $v$ and $v^*$.  Let $u$ be the endpoint of the unique
  edge $e \in\OPT \cap C$ that is on the same side of the cut as $v$
  and $v^*$.  Then because $e$ is the only edge in $\OPT\cap C$ and
  because $C$ separates $v$ and $v^*$ from $w$ and $w^*$, both paths
  $R$ and $R^*$ must contain the edge $e$.  Moreover, the $v$-$u$
  subpath of $R$ and the $v^*$-$u$ subpath of $R^*$ do not contain any
  edge of the cut $C$.  Because $C$ was chosen as the closest cut to
  $v$ and $v^*$ among all cuts induced by $\mathcal{C}$ that contain
  only one edge from $\OPT$, we conclude that neither of these two
  subpaths contains an edge of $D^*$.  Hence, both $v$ and $v^*$ are
  connected to $u$ in $(R\cup R^*)\setminus D^*$, which implies that
  $(R\cup R^*)\setminus D^*$ contains a $v$-$v^*$ path. This completes
  the proof of the claim, and hence of \Cref{lem:autarkic_pairs}.
\end{proof}

Now combining \Cref{lem:autarkic_pairs} with a $2$-approximation
algorithm for the contracted instance, and recalling from
\eqref{eq:trivial_lower_bound} that the rightmost sum in~(\ref{eq:1})
is denoted by $\Lambda$, which is at most $\excess(\OPT)$, we get:
\begin{corollary}[Profit from Autarkic Collections]
  \label{cor:autarkic_pairs}
  For a crossing-free autarkic collection $\cP$, there is an
  efficient algorithm to find a Steiner forest of cost 
  \begin{gather}
 2\, c(\OPT) ~~-~~\underbrace{(2\,\coverage(\cP) - \cost(\cP))}_{=
      \profit(\cP)} ~~+~~ 2\excess(\OPT). \label{eq:2}
  \end{gather}
\end{corollary}

If $\excess(\OPT)$ is large, we get an improved lower bound on $c(\OPT)$, and hence a better-than-two
approximation ratio even for the $\eps$-extended algorithm. 
Otherwise, the last summand is negligible; in this case, we gain if
$\profit(\cP)$ is large. 
We now prove that we can efficiently find a crossing-free autarkic collection with maximum profit,
before we explicitly state the final algorithm for Steiner Forest.

\subsection{Finding a Maximum-Profit Autarkic Collection}
\label{sec:aut-algo}

In this section we prove \Cref{thm:aut-algo}. 
To this end, we first show that the set of autarkic pairs has a ``laminar"
structure.  To formalize this, we introduce the following partial
order on the autarkic tuples:

\begin{definition}
  Let $P$ and $Q$ be autarkic tuples with $P=\{A_1,A_2\}$ or
  $P=\{A_1,A_2, A_3\}$, and $Q=\{B_1,B_2\}$ or $Q=\{B_1,B_2,B_3\}$.
  Then we write $P\subseteq Q$ if every set $A_i$ is contained in some
  set $B_j$.  Moreover, we write $P\cap Q = \emptyset$ if 
  $A_i \cap B_j = \emptyset$ for all $i,j$.
\end{definition}

Note that the $\subseteq$ relation is transitive: if
$P_1 \subseteq P_2$ and $P_2 \subseteq P_3$ for three autarkic tuples
$P_1, P_2, P_3$, then we also have $P_1 \subseteq P_3$.  Moreover, if
$P_1 \subseteq P_2$ and $P_2 \cap P_3 =\emptyset$, then
$P_1 \cap P_3 = \emptyset$.  Finally, $P_1 \subseteq P_2$ and
$P_2 \subseteq P_1$ implies $P_1 = P_2$.

Now we can observe that the set of autarkic pairs indeed has a laminar structure.

\begin{lemma}\label{lem:aut-laminar}
  Let $P$ and $Q$ be autarkic tuples.  Then we have $P\subseteq Q$, or
  $Q\subseteq P$, or $P\cap Q = \emptyset$.
\end{lemma}
\begin{proof}
  Let $P=\{A_1,A_2\}$ or $P=\{A_1,A_2, A_3\}$, and $Q=\{B_1,B_2\}$ or
  $Q=\{B_1,B_2,B_3\}$.  By the definition of an autarkic tuple there
  exist times $t(P)$ and $t(Q)$, such that all sets $A_i$ are active
  at time $t(P)$ and all sets $B_i$ are active at time $t(Q)$.
  Without loss of generality, we have $t(P) \leq t(Q)$.  Then for
  every set $A_i$ and every set $B_j$, we have either
  $A_i \cap B_j = \emptyset$ or $A_i \subseteq B_j$.

  If all sets $A_i$ are disjoint from all sets $B_j$, we have
  $P\cap Q=\emptyset$.  Moreover, if every set $A_i$ is contained in
  some set $B_j$, we have $P\subseteq Q$.  So assume that neither of
  these is the case.  Then, without loss of generality,
  $A_1 \subseteq B_1$ and $A_2$ is disjoint from all sets $B_j$.

  We claim that there is some demand pair $\{v,w\}$ such that $v$ is
  contained in some set $B_j$ and $w$ is contained in none of the sets
  $B_j$.  From this claim, we derive a contradiction as follows.  If
  $Q$ is an autarkic pair, the claim implies that $\{v,w\}$ belongs to
  exactly one of $\sep(B_1)$ or $\sep(B_2)$ and not to both of them,
  contradicting $\sep(B_1)=\sep(B_2)$.  Otherwise, $Q$ is an autarkic
  triple and $\{v,w\}$ belongs to $\sep(B_1\cup B_2 \cup B_3)$,
  contradicting $\sep(B_1\cup B_2 \cup B_3)=\emptyset$.

  Now it remains to show the existence of the claimed demand pair
  $\{v,w\}$.  First, consider the case that $P$ is an autarkic
  \emph{pair}.  Since $\sep(A_1) =\sep(A_2) \neq \emptyset$, there is
  a demand pair $\{v,w\}$ with $v\in A_1 \subseteq B_1$ and
  $w\in A_2$.  Because $A_2$ is disjoint from all sets $B_j$, the
  vertex $w$ is indeed not contained in any set $B_j$.

  Next, we consider the remaining case that $P$ is an autarkic
  \emph{triple}.  Suppose $A_3$ is also disjoint from all sets $B_j$:
  then let $\{v,w\}$ be some demand pair in the non-empty set
  $\sep(A_1)$ with $v\in A_1$.  Then because
  $\sep(A_1\cup A_2\cup A_3) =\emptyset$, we have $w\in A_2\cup A_3$
  and thus $w$ is indeed not contained in any set $B_j$.  Otherwise,
  $A_3$ is contained in some set $B_j$ and we let $\{v,w\}$ be some
  demand pair in the non-empty set $\sep(A_2)$ with $w\in A_2$.  In
  particular, $w$ is not contained in any set $B_j$.  Moreover,
  because $\sep(A_1\cup A_2\cup A_3) =\emptyset$, we have
  $v\in A_1\cup A_3$ and thus $v$ is contained in some set $B_j$.
\end{proof}

\begin{lemma}\label{lem:crossing-monotone}
  Let $P_1,P_2,P_3$ be autarkic tuples, with 
  $P_1\subseteq P_2 \subseteq P_3$. Suppose $\sep(P_1) \cap \sep(P_3)
  \neq \emptyset$. Then $\sep(P_2) \cap \sep(P_3) \neq \emptyset$.
\end{lemma}
\begin{proof}
  Let $P_3 = \{B_1, B_2\}$ or $P_3=\{B_1, B_2, B_3\}$.  Because
  $P_1 \subseteq P_3$, each set of the autarkic tuple $P_1$ is a
  subset of some $B_j$.  Moreover, since
  $\sep(P_1) \cap \sep(P_3) \neq \emptyset$, not all of the sets in
  $P_1$ can be a subset of the same $B_j$. Thus, using
  $P_1 \subseteq P_2$, we conclude that not all sets from the autarkic
  tuple $P_2$ can be subsets of the same $B_j$.

  Let $P_2=\{A_1, A_2\}$ or $P_2=\{A_1,A_2,A_3\}$.  Then, in
  particular, there is one set $A_i$ of the tuple $P_2$ that must be
  contained in a different set $B_j$ than the other sets from $P_2$.
  Consider a demand pair $\{v,w\}$ in the non-empty set $\sep(A_i)$,
  where $v$ is contained in this set $A_i$.  Then $v$ and $w$ belong
  to different sets $B_j$ and thus $\{v,w\} \in \sep(P_3)$.  Because
  $\{v,w\}$ is also contained in $\sep(A_i) \subseteq \sep(P_2)$, this
  implies that $\sep(P_2) \cap \sep(P_3) \neq \emptyset$.
\end{proof}

We can now give the algorithm to find the profit-maximizing autarkic
collection. We restate \Cref{thm:aut-algo} here for convenience:

\EffAut*

\begin{proof}[Proof of \Cref{thm:aut-algo}]
  If $\mathcal{P}$ is a crossing-free autarkic collection, then for
  every set $U\in \supp(y)$, there is at most one tuple $P\in \cP$
  covering $U$.  Thus,
  $\coverage(\cP) = \sum_{P\in \cP} \coverage(\{P\})$.  Because we
  also have $\cost(\cP) = \sum_{P\in \cP} \cost(\{P\})$, this implies
  $\profit(\cP) = \sum_{P\in \cP} \profit(\{P\})$.  We conclude that it
  suffices to find a crossing-free autarkic collection $\cP$
  maximizing the sum of the profits $\profit(\{P\})$ with $P\in \cP$.

  Because there is only a polynomial number of autarkic tuples, we can
  list all of them in polynomial time.  Moreover, we can compute
  $\profit(\{P\})$ for each autarkic tuple.  Let $\cQ$ be the set of
  all autarkic tuples.  We show how to use the laminarity to
  efficiently compute a crossing-free collection $\cP \subseteq \cQ$
  with maximum profit using a dynamic program.  In our dynamic
  program, for each autarkic tuple $Q\in \cQ$, we compute a
  crossing-free autarkic collection
  \[
    \cP(Q) \subseteq \big\{ P\in \cQ : P \subseteq Q\big\}
  \]
  with maximum profit.  We compute these collections $\cP(Q)$ in an
  order such that for $Q_1 \subseteq Q_2$, we compute $\cP(Q_1)$
  before $\cP(Q_2)$.  Thus, to complete the proof, it suffices to show
  that we can compute $\cP(Q)$ in polynomial time, assuming that we
  are given the collections $\cP(P)$ for all $P\in \cQ$ with
  $P \subset Q$.

  To compute $\cP(Q)$, we compute a maximum-profit solution for the
  two cases ($Q\in \cP(Q)$ or not) and take a solution with maximum
  profit among these two.  If $Q\notin \cP(Q)$, we let
  $P_1, \dots, P_k$ be the autarkic tuples with $P_i\subseteq Q$ that
  are maximal with respect to the partial order $\subseteq$.  By
  \Cref{lem:aut-laminar}, we have $P_i \cap P_j =\emptyset$ for all
  $i\neq j$.  Thus, the sets $\cP(P_i)$ are disjoint and their union
  $\cP(P_1) \cup \cP(P_2) \cup \dots \cup \cP(P_k)$ is crossing-free.
  Therefore, by the profit-maximality of the sets $\cP(P_i)$, this
  union is a maximum-profit subset of
  $\{ P\in \cQ : P \subseteq Q \} \setminus \{Q\}$.

  Now it remains to consider the case $Q\in \cP(Q)$.  Among all
  autarkic tuples $P$ that are not crossing $Q$ (i.e., tuples $P$ such
  that $\sep(P) \cap \sep(Q) = \emptyset$) and satisfy $P\subseteq Q$,
  let $P_1, \dots, P_k$ be those autarkic tuples that are maximal with
  respect to the partial order $\subseteq$.  We claim that
  \[ \cP \coloneqq \{Q\} \cup \cP(P_1) \cup \cP(P_2) \cup \dots \cup
  \cP(P_k) \] is crossing-free and has maximum profit among all
  crossing-free subsets of
  $\{ P\in \cQ : P \subseteq Q \}$ %
  containing $Q$.
  Note that by \Cref{lem:aut-laminar}, we have
  $P_i \cap P_j =\emptyset$ for all $i\neq j$ and thus in particular,
  the sets $\cP(P_i)$ are disjoint.  Thus, the maximality of the
  profit follows directly from the profit-maximality of the sets
  $\cP(P_i)$.  To show that $\cP$ is crossing-free, observe that
  because $P_i \cap P_j =\emptyset$ for all $i\neq j$, the set
  $\cP(P_1) \cup \cP(P_2) \cup \dots \cup \cP(P_k)$ is crossing-free.
  Suppose for the sake of a contradiction that $Q$
  crosses some $P \in \cP(P_i)$.  Then $P \subseteq P_i \subseteq Q$
  and thus by \Cref{lem:crossing-monotone}, the autarkic tuple $P_i$
  must cross $Q$, which contradicts the choice of the tuples $P_i$.
  This completes the proof of correctness for the dynamic program, and
  hence the proof of~\Cref{thm:aut-algo}.
\end{proof}

\section{The Algorithm for Steiner Forest}
\label{sec:main-algo}

Given the machinery developed in the preceding sections, the algorithm
for Steiner forest can be compactly described:

\lineshere
\begin{enumerate}[label=(\arabic*)]
\item \label{step:1} Run the $\eps$-extended moat-growing algorithm to
  define the active components $\cA^t$ and dual values
  $\{y_S\}_{S \sse V}$. Use \Cref{thm:AKR-GW} to get a solution $F_1$.

\item Start with the duals from Step~\ref{step:1} and use
  \Cref{thm:effic-contraction} %
  get a
  solution $F_2$.

\item Again, start with the duals from Step~\ref{step:1}, and use
  \Cref{thm:aut-algo} to find a maximum-profit autarkic collection
  $\cP$. Buy the edges in the connectors of tuples in $\cP$, and run a
  $2$-approximation algorithm on the contracted instance
  $\cI/\cP$. This gives the solution $F_3$.

\item Among these three solutions, return the solution with least cost. 
\end{enumerate}

\lineshere

We show that, given the dual solution $y$:
\begin{theorem}\label{thm:win-win-win}
  There exists a collection $\cS$ of actively connected vertex sets,
  and a crossing-free autarkic collection $\cP$ such that
  \[
    \frac{10}{3}\cdot (\gain(\cS) - \cost(\cS)) + (2\coverage(\cP) -
    \cost(\cP)) + \Big(\frac{1}{2\epsilon}+17\Big)
    \cdot \excess(\OPT) \ \geq\ y(\Usep).
  \]
\end{theorem}
In other words, (i)~either the $\excess(\OPT)$ is large, or (ii)~there
exists a collection $\cS^*$ of actively connected vertex sets with
large gain, or (iii)~there exists an autarkic collection $\cP^*$ with
large profit. These three events are precisely what allow us to show
that the best of $F_1$, $F_2$, and $F_3$ is a $\appx$-approximation to
Steiner Forest.

Two asides: to get the best approximation guarantee, we don't just use
the tradeoff from \Cref{thm:win-win-win}: we use the helper theorems
involved in its proof. Moreover, we can show that just taking the
better among forests $F_2$ and $F_3$ suffices (since the forest $F_1$
is essentially what we get as $F_2$ with $\cS$ being empty); however, we retain
that step because it emphasizes the role of $\excess(\OPT)$.

We prove \Cref{thm:win-win-win} in \Cref{sec:satisfied_pairs_or_large_excess}, and then
combine the various bounds together in \Cref{sec:approximation-ratio}
to derive the final approximation factor.

\section{Obtaining Large Gain or Large Profit}
\label{sec:satisfied_pairs_or_large_excess}

In this section, we prove the ``win-win-win'' guarantee of
\Cref{thm:win-win-win}. In \Cref{sec:canon-collection}, we give a
construction of a ``canonical'' collection $\cS$ of actively connected
sets, and prove some properties for it. Then, in
\Cref{sec:coverage_or_gain}, we show how to construct an autarkic
collection $\cP$ with large profit (i.e., its coverage minus its
cost) when the ``profit'' of $\cS$ (i.e., its gain minus its own cost)
is not very large.

\subsection{Constructing the Canonical Collection $\cS$ from $\OPT$}
\label{sec:canon-collection}

In order to construct the ``canonical'' collection $\cS$ of actively
connected vertex sets, we map every set $U\in \Usep$ to a demand pair
$\varphi(U)$ separated by $U$, and refer to it as the \emph{demand
  pair responsible for $U$}. In case $U$ separates several demand
pairs, we assign an arbitrary one. This mapping allows us to partition
the total value of $\Usep$ among the demand pairs.

\begin{definition}[Value and Satisfaction]
  For a demand pair $d \in \cD$, we define its \emph{value} as the
  total dual value it is responsible for; i.e., %
  \[
    \val(d)\ \coloneqq \ \sum_{U: \varphi(U)=d} y_U.
  \]
  
  A demand pair $d=\{a,b\}$ is \emph{satisfied} by a vertex set $S$ if
  $\{a,b\}\sse S$; $d$ is satisfied by $\cS$ if it is satisfied by
  some set $S \in \cS$. For a collection $\cS$ of vertex sets, we
  define its value as
  \[
    \val(\cS) \ \coloneqq\ \sum_{\substack{d\text{ satisfied}\\\text{by some }S\in \cS}} \val(d).
  \]
\end{definition}

\begin{definition}[Agreeable Collection]
  A collection of vertex sets $\cS$ is called
  \emph{agreeable} (with respect to $y$) if
  \begin{enumerate}[label=(A\arabic*)]
  \item\label{item:property_construction_from_opt} every set
    $S \in \cS$ is a subset of the vertex set of some connected
    component of OPT, and there is at most one such set $S \in \cS$
    for each connected component of OPT, and 
  \item\label{item:property_all_or_nothing} if demand pairs $\{a,b\}$
    and $\{a', b'\}$ are both separated by some set $U$ in the support
    of $y$, then either none or both of these demands are satisfied by
    $\cS$.
  \end{enumerate}
\end{definition}

\subsubsection{Defining the Canonical Collection}
\label{sec:definition-of-cS}

Consider the connected components of $\OPT$. We partition it to obtain
a collection $\mathcal{F}$ of forests as follows. Start with each tree
of $\OPT$ being its own (single-treed) forest in $\cF$. Now, for each
set $U\in \Usep$, merge the forests from $\mathcal{F}$ connecting some
demand pair separated by $U$ into a single forest.  (Because the
forests in $\mathcal{F}$ are always vertex disjoint, merging different
forests in $\cF$ simply involves taking the union of their vertex sets
and the union of their edge sets.)

Next, we describe our construction of the collection $\cS$.  We start
with $\cS = \emptyset$.  For each forest $F\in \mathcal{F}$, let $r_F$
be a vertex of $F$ with the maximum deactivation time.  For each
connected component $T$ of the forest $F$, let $S$ be the set of all
vertices in $T$ which are actively connected to $r_F$. If
$S \neq \emptyset$, add $S$ to $\cS$.  Finally, define $T_S$ to be the
minimal subtree of the optimal tree $T$ connecting $S$.

\subsubsection{Properties of the Canonical Collection}
\label{sec:properties-of-cS}

\begin{lemma}\label{lem:constructing_trees_from_opt_part_1}
  The collection $\cS$ of actively connected vertex sets defined above satisfies:
 \begin{enumerate}[nosep,label = {(\roman*)}]
 \item \label{item:property_tree_connects_s} every tree $T_S$ connects the vertices in $S$,
 \item \label{item:property_S_agreeable} $\cS$ is agreeable with
   respect to $y$.
 \end{enumerate}
\end{lemma}

\begin{proof}
  Observe that every set $S\in \cS$ is actively connected and
  satisfies property~\ref{item:property_tree_connects_s} by
  construction. Property~\ref{item:property_construction_from_opt} of
  agreeable collections also follows by construction, and it remains
  to show property~\ref{item:property_all_or_nothing}. Consider a set
  $U\in \Usep$ and demand pairs $d = \{a,b\}$ and $d' = \{a', b'\}$
  that are separated by $U$ (with say, $a,a' \in U$). Observe that the optimal
  trees connecting $d,d'$ must belong to the same forest $F$ of $\cF$,
  by construction. Moreover, the set $\{a,b,a',b'\}$ belongs to the
  same equivalence class of actively connected vertices: $\{a,a'\}$
  because they both lie in $U$ while both are still active,
  $\{a,b\}, \{a',b'\}$ because any demand pair must be actively
  connected. Now if $r_F$ belongs to the same equivalence class, then
  both demands will be satisfied by $\cS$, else neither will be,
  proving property~\ref{item:property_all_or_nothing}.
\end{proof}

\begin{lemma}[Cheap Trees]
  \label{lem:constructing_trees_from_opt_part_2}
  The collection $\cS$ satisfies:
   \begin{equation}\label{eq:cheap_trees}
     \sum_{S\in \cS} c(T_S) \ \leq \ \val(\cS) + 2 \cdot \excess(\OPT).
   \end{equation}
\end{lemma}

\begin{proof}
  For a connected component $T$ of OPT, define its excess as
  \begin{gather}
    \excess(T) \coloneqq c(T) - \val(\{V(T)\}) \ \geq\ c(T) - \sum_{U:
      T \cap \delta(U) \neq \emptyset} y_U \ \geq\ 0,
  \end{gather}
  where the last inequality holds because $y$ satisfies the dual LP
  constraints.  Using that every demand pair is satisfied by exactly
  one connected component of OPT, and that the value function assigns
  each set $U \in \Usep$ to exactly one demand pair, we get
  \begin{align*}
    \excess(\OPT)\ 
    =&\  c(\OPT) - \sum_{U\in \Usep} y_U  \\
    =&\ c(\OPT) - \val\big(\{ V(T) : T\text{ connected comp. of OPT}\}\big) \\
    =&\ \sum_{\substack{T:\text{ conn. comp.}\\\text{of OPT}}} \Big( c(T) - \val(\{V(T)\}) \Big) \\
    =&\ \sum_{\substack{T:\text{ conn. comp.}\\\text{of OPT}}} \excess(T).
  \end{align*}
  Because the sets in $\cS$ are disjoint, we have
  $\val(\cS) = \sum_{S\in \cS} \val(\{S\})$.  Thus, it suffices to
  show that for any single connected component $T$ of $\OPT$ and the
  unique set $S\in \cS$ corresponding to $T$ (if such a set $S$
  exists), we have
  \begin{equation}\label{eq:cost_bound_for_single_tree}
    c(T_S) \ \leq\ \val(\{S\}) + 2\cdot \excess(T).
  \end{equation}

  To prove~\eqref{eq:cost_bound_for_single_tree}, we relate the
  difference $c(T) - c(T_S)$ with $\val(\{V(T)\}) -
  \val(\{S\})$. Indeed, let $\mathcal{U}_{\rm loss} \subseteq \Usep$
  be the set of all $U \in \Usep$ for which the demand pair
  $\varphi(U)$ is satisfied by $V(T)$ but not by $S$.  Then we have
  \begin{equation}\label{eq:value_loss}
    y(\mathcal{U}_{\rm loss}) = \val(\{V(T)\}) - \val(\{S\}). 
  \end{equation}

  Next, we establish a lower bound on $c(T) - c(T_S)$.  Consider a set
  $U\in \mathcal{U}_{\rm loss}$. We show that $S\cap U=
  \emptyset$. Say $\varphi(U) = \{a,b\}$ with $a \in U$. By definition
  of $\mathcal{U}_{\rm loss}$, the demand pair $\varphi(U)$ is not
  satisfied by $S$.  Because the two vertices of a demand pair are
  always actively connected to each other, by construction of $S$,
  neither of $\{a,b\}$ is actively connected to the vertex $r_F$ (for
  the forest $F\in \mathcal{F}$ that has $T$ as a connected
  component).  Because $U \in \supp(y)$, every vertex $w\in U$ is
  either actively connected to $a$ or has an deactivation time
  strictly smaller than $a$'s.  In the former case, $w$ lies in the
  same equivalence class of actively connected components as $a$, and
  hence cannot be actively connected to $r_F$. In the latter case, $w$
  has strictly smaller deactivation time than $a$, and thus strictly
  than $r_F$ (since $r_F$ was chosen to have the largest deactivation
  time among all vertices in $F$).  In both cases, we conclude that
  $w \in U$ is not actively connected to $r_F$ and hence not contained
  in $S$.  This shows $S\cap U = \emptyset$.

  This claim means that if $T\cap \delta(U) = 1$, then the unique edge
  in $T\cap \delta(U)$ cannot be connecting two vertices in $S$, and
  is not part of the tree $T_S$.  Using that $y$ satisfies the dual LP
  constraints, we get
  \begin{equation}\label{eq:cost_decrease}
    \sum_{\substack{U \in \mathcal{U}_{\rm loss}:\\ |T \cap
        \delta(U)|= 1 }} y_U\ \leq\ \sum_{U: (T\setminus T_S) \cap
      \delta(U)\neq\emptyset} y_U \ \leq\ c(T\setminus T_S) = 
    c(T) - c(T_S).
  \end{equation}
  Combining
  \eqref{eq:value_loss} and \eqref{eq:cost_decrease}, we obtain
  \begin{align}
    c(T_S) - \val(\{S\})\ 
    \leq &\ c(T) - \sum_{\substack{U \in \mathcal{U}_{\rm loss}:\\ |T
    \cap \delta(U)|= 1 }} y_U - \val(\{V(T)\}) + y(\mathcal{U}_{\rm
    loss})  \notag \\
    =&\ \excess(T) + \sum_{\substack{U \in \mathcal{U}_{\rm
       loss}:\\ |T \cap \delta(U)|> 1 }} y_U \label{eq:4} \\
    \leq&\ 2\cdot \excess(T); \notag
  \end{align}
  in equality~(\ref{eq:4}) we used that $|T \cap \delta(U)| \geq 1$ for all
  $U \in \mathcal{U}_{\rm loss}$, because $T$ satisfies the demand
  pair $\varphi(U)$ separated by $U$. This allowed to partition
  $y(\cU_{\rm loss})$ into sums over those $U$'s which are crossed by
  $T$ exactly once, and more than once.  This shows
  \eqref{eq:cost_bound_for_single_tree} and completes the proof of
  \Cref{lem:constructing_trees_from_opt_part_2}.
\end{proof}

\begin{lemma}[Non-Trivial Excess]
  \label{lem:constructing_trees_from_opt_part_3}
  The collection $\cS$ satisfies:
   \begin{equation}\label{eq:large_excess}
     \excess(\OPT) \ \geq \ 2\epsilon \cdot \Big( y(\Usep) - \val(\cS) \Big) .
   \end{equation}
\end{lemma}

\begin{proof}
  It suffices to show that
  \begin{equation}\label{eq:lower_bound_nonsep}
    \sum_{\substack{W\in \mathcal{U}_{\rm unsep}:\\ \OPT\cap \delta(W) \neq \emptyset }} y_W \ \geq\  \epsilon \cdot \big( y(\Usep) - \val(\cS)\big).
  \end{equation}
  Indeed, if $\OPT\cap \delta(W)$ contains only a single edge, then
  because $W$ does not separate any demand pair, we could remove this
  edge and maintain a feasible solution, contradicting the minimality
  of OPT.  Therefore, for every $W\in \mathcal{U}_{\rm unsep}$, we
  have $\OPT\cap \delta(W) = \emptyset$ or
  $|\OPT\cap \delta(W)| \geq 2$.  The definition of excess combined
  with \eqref{eq:lower_bound_nonsep} implies
  \[
    \excess(\OPT) \geq \sum_{W\in \mathcal{U}_{\rm unsep}} |\OPT\cap \delta(W)| \cdot y_W 
    \ \geq\ 2 \cdot \sum_{\substack{W\in \mathcal{U}_{\rm unsep}:\\ \OPT\cap \delta(W) \neq \emptyset }} y_W
    \ \geq\ 2\epsilon \cdot \big( y(\Usep) - \val(\cS)\big).
  \]
  In the rest of the proof, we show~\eqref{eq:lower_bound_nonsep}. 

  We consider a set $U$ from $\Usep$.  Because $U\in \Usep$, the
  moat-growing process generates $\epsilon \cdot y_U$ money when
  growing $U$. Consider some set $W\in \mathcal{U}_{\rm unsep}$ whose
  growth was (partially) paid by some money collected by $U$.  We show
  that either $\varphi(U)$ is satisfied by $\cS$, i.e., $U$
  contributes to $\val(\cS)$, or $\OPT \cap \delta(W) \neq \emptyset$.
  Note that this implies \eqref{eq:lower_bound_nonsep}.

  Let $F\in \mathcal{F}$ be the forest satisfying the demand pair
  $\varphi(U)$; such a forest exists because $\mathcal{F}$ partitions
  $\OPT$. If $\OPT \cap \delta(W) = \emptyset$, every connected
  component of $F$ is either completely inside $W$, or completely
  outside of $W$. Because $U$ pays for $W$, we have $U\subseteq W$,
  and hence at least one connected component of $F$ is inside $W$.  We
  now distinguish two cases.

  The first case is where some connected component of $F$ is outside
  $W$.  Then, by construction of $\mathcal{F}$, there must be some set
  $\widehat{U}\in \Usep$ and two connected components $T_1$ and $T_2$
  of $F$ such that
  \begin{itemize}
  \item $T_1$ is inside $W$ and $T_2$ is outside $W$, and
  \item each of the trees $T_1$ and $T_2$ satisfies some demand pair separated by $\widehat{U}$.
  \end{itemize}
  This implies $\widehat{U}\cap W \neq \emptyset$ and
  $\widehat{U}\cap (V\setminus W) \neq \emptyset$. Now since both $W$
  and $\widehat{U}$ belong to the support of $y$, this contradicts the
  laminarity of the support of $y$.

  The second case is where the forest $F$ is completely inside $W$.
  We now show that $\varphi(U)$ is satisfied by $\cS$, i.e., $U$
  contributes to $\val(\cS)$.  Let $u\in U$ be the vertex of the
  demand pair $\varphi(U)$ that is contained in $U$.  Because $U$
  separates $\varphi(U)$, the vertex $u$ was still active at any time
  when $y_U$ was growing.  Because $U$ paid for $W$, the vertex $u$
  remained active at least until the time when $y_W$ was growing,
  using money generated by $U$.  Thus, by the choice of the vertex
  $r_F$, also $r_F$ was still active when $y_W$ was growing using
  money generated by $U$.  Using $u\in U \subseteq W$ and
  $r_F \in W$, we conclude that $r_F$ and $u$ are actively connected,
  i.e., $u\in S$ for the set $S$ corresponding to the connected
  component $T$ of $F$ that contains $u$.  Because the elements of
  every demand pair are actively connected, this implies that
  $\varphi(U)$ is satisfied by $S$.
\end{proof}

\subsection{Constructing the autarkic collection $\cP$ from $\OPT$}
\label{sec:coverage_or_gain}

Having constructed the canonical collection $\cS$ of actively
connected sets in the preceding section, we now show that if $\cS$ has
small gain relative to its cost, and the excess is small, there is an
autarkic collection $\cP$ with large profit.

\begin{lemma}\label{lem:win-win}
  Let $\cS$ be the agreeable collection of actively connected vertex
  sets constructed in \Cref{sec:canon-collection}.  Then there exists
  a crossing-free autarkic collection $\cP$ such that
  \begin{equation}\label{eq:win-win}
    \frac{10}{3}\cdot (\gain(\cS) - \cost(\cS)) + \profit(\cP) +
    17 \cdot \excess(\OPT) \geq \val(\cS).
  \end{equation}
\end{lemma}

Let $\cU_{\cS}$ be the collection of sets contributing to $\val(\cS)$,
that is
\[
\cU_{\cS} \ \coloneqq\ \{ U \in \supp(y) \cap \Usep : \cS\text{ satisfies }\varphi(U)\}.
\]
and for $S\in \cS$, let $\cU_S \coloneqq \cU_{\{S\}}$ be the
collection of sets contributing to $\val(\{S\})$. For any time $t$,
let $\cU_{\cS}^t := \cU_{\cS} \cap \cA^t$ and $\cU_{S}^t := \cU_{S}
\cap \cA^t$ be the corresponding sets that are active at time $t$.
Then
\begin{equation}\label{eq:ww:simple_expression_value}
\val(\cS) \ =\ \int_{t=0}^{\infty} |\cU_{\cS}^t | \ dt \ =\  \sum_{S
  \in \cS} \int_{t=0}^{\infty} |\cU_{S}^t | \ dt.
\end{equation}

For a time $t$ and a set $U\in \cU_{\cS}^t$, let
$m(U,t) \geq 0$ be the number of sets from $ \cU_{\cS}^t$
which with the set $U$ is merged when going from $\cA^t$ to
$\cA^t/\cS$.  Because the demand pair $\varphi(U)$ is separated by
$U$, there is some set in $\cA^t\setminus \{U\}$ that is also
separated by $\varphi(U)$.  Since $\cS$ satisfies the demand pair
$\varphi(U)$, the set $U$ is merged with at least one other set when
going from $\cA^t$ to $\cA^t/\cS$, and so $m(U,t) \geq 1$.  We observe
that
\begin{equation}\label{eq:gain_simple_bound}
  {\rm gain}(\cS)\ = \ 2 \cdot \int_{t=0}^{\infty} \big( |\cA^t|
  -|\cA^t / \cS| \big)\  dt 
  \ \geq \ 2\cdot \int_{t=0}^{\infty}\sum_{U\in \cU_{\cS}^t} \frac{m(U,t)}{m(U,t)+1} \  dt.
\end{equation}

Recall from \Cref{lem:constructing_trees_from_opt_part_2} that the
cost term is approximately equal to $\val(\cS)$. Consequently, if
$m(U,t) \geq 2$---i.e., when each set is being merged with at least
two others---most of the time, we get a contribution from $\gain(\cS)$
that is much more than $\cost(\cS)$. In fact, the larger the value of
$m(U,t)$, the easier the proof. The tricky cases are when
$m(U,t) = 1$, but then it seems natural to extract out autarkic pairs,
which is exactly what the rest of the construction does.

\begin{definition}[$S$-Good/Bad Times]
  For a set $S\in \cS$, a time $t\geq 0$ is called \emph{$S$-bad} if
  some set $U\in \cA^t$ that separates $S$ also separates some other
  set $S'\in \cS\setminus\{S\}$; otherwise, the time $t$ is
  \emph{$S$-good}.
\end{definition}

The following lemma formalizes the intuition that a small excess means
that bad times are rare; in fact, for the first read, one could
imagine all times being good. (Their presence just adds complexity to
the proof, and increases the constant in front of the $\excess(\OPT)$
term.)

\begin{lemma}[Few Bad Times]
  \label{lem:bad_times_yield_excess}
  We have 
  \[
    B(\cS) \ \coloneqq\ \int_{t=0}^{\infty}  |\{ S \in \cS : t \text{ is $S$-bad}\}| \ dt \ \leq\ 2 \cdot \excess(\OPT).
  \]
\end{lemma}
\begin{proof}
  Using the dual feasibility of $y$ and $|\OPT \cap U|\geq 1$ for all
  $U\in \Usep$, we have
  \begin{align*}
    \excess(\OPT) 
    \ = &\ c(\OPT) - \sum_{U\in \Usep} y_U 
          \ \geq \ \sum_{U\in \supp(y)} y_U \cdot |\OPT\cap \delta(U)|  - \sum_{U\in \Usep} y_U \\
    \geq&\  \sum_{U\in \supp(y)} y_U \cdot \max\big\{ |\OPT\cap \delta(U)| - 1, 0\big\} \\
    =&\ \int_{t=0}^{\infty} \sum_{U\in \cA^t} \max\big\{ |\OPT\cap \delta(U)| - 1, 0\big\}.
  \end{align*}
  Thus, it suffices to show that for every time $t\geq 0$, we have
  \begin{equation}\label{eq:bad_events}
    |\{ S \in \cS : t \text{ is $S$-bad}\}| \ \leq\ 2
    \cdot \sum_{U\in \cA^t} \max\big\{ |\OPT\cap \delta(U)| - 1,
    0\big\}. 
  \end{equation}
  Call a set $U\in \cA^t$ \emph{double-crossing} if
  $|\OPT\cap \delta(U)| > 1$. The collection $\cS$ is agreeable, so
  property~\ref{item:property_construction_from_opt} means that each
  $S \in \cS$ is a subset of some tree in $\OPT$.  Thus, if some set
  $U\in \supp(y)$ separates some $k$ sets from $\cS$, then we have
  $|\OPT\cap \delta(U)| \geq k$. If $U$ is double-crossing, it can
  give out $k \geq 2$ tokens, and each set $S \in \cS$ separated by it
  can collect one such token. By the definition of $S$-bad times,
  every set $S$ for which the time $t$ is $S$-bad then receives at
  least one token. Moreover, since $k \geq 2$,
  $k \leq 2\, \max(k- 1, 0)$, which implies~\eqref{eq:bad_events}, and
  completes the proof.
\end{proof}

\subsubsection{Lower-Bounding the Benefit from $\cS$}
\label{sec:lower-bound-benef}

In order to prove \Cref{lem:win-win}, we separate out $\val(\cS)$
depending on how many active sets are being merged by the sets
$S \in \cS$ (whether the number is $2, 3$, or higher)---since these
correspond to the gain we accrue---and whether the time is good or bad
for these sets---since this corresponds to the amount of
double-counting we have to account for. 
We prove bounds on $\gain(\cS)$ and $\cost(\cS)$, as well as on
$\coverage(\cP)$ and $\cost(\cP)$ in terms of these fine-grained
quantities.  Summing up a suitable linear combination of these bounds
then yields~\eqref{eq:win-win}.  We define
\begin{align*}
  x_1 \ \coloneqq&\ \int_{t=0}^\infty \sum_{\substack{S\in \cS: \\ t\text{ is $S$-good}}}   |\{ U\in \cU_{S}^t : m(U,t) = 1\}|   \ dt, \\
  \overline{x}_1 \ \coloneqq&\ \int_{t=0}^\infty \sum_{\substack{S\in \cS: \\ t\text{ is $S$-bad}}}   |\{ U\in \cU_{S}^t: m(U,t) = 1\}|   \ dt, \\
  x_2 \ \coloneqq&\ \int_{t=0}^\infty \sum_{\substack{S\in \cS: \\ t\text{ is $S$-good}}}   |\{ U\in \cU_{S}^t: m(U,t) = 2\}|   \ dt, \\
  \overline{x}_2 \ \coloneqq&\ \int_{t=0}^\infty \sum_{\substack{S\in
                              \cS: \\ t\text{ is $S$-bad}}}   |\{ U\in
  \cU_{S}^t : m(U,t) = 2\}|   \ dt, \quad\text{and}  \\
  x_3 \ \coloneqq&\ \int_{t=0}^\infty   |\{ U\in \cU_{\cS}^t: m(U,t) \geq 3 \}|   \ dt .
\end{align*}
Using \eqref{eq:ww:simple_expression_value} and the fact that the sets $\cU_S$ with $S\in \cS$ form a partition of $\cU_{\cS}$, we then get
\begin{equation}\label{eq:bound_value}
\val(\cS)\ =\ x_1 + \overline{x}_1 + x_2 + \overline{x}_2 + x_3.
\end{equation}

\begin{lemma}\label{lem:bound_gain_cost}
  ${\rm gain}(\cS) - \cost(\cS) \ \geq\   \tfrac{1}{3} \cdot (x_2 + \overline{x}_2) + \tfrac{1}{2} \cdot x_3 - 2\cdot \excess(\OPT).$
\end{lemma}
\begin{proof}
  Using~\eqref{eq:gain_simple_bound} and the definition of $x_1,\overline{x}_1, x_2, \overline{x}_2,x_3$, we have
  \begin{align*}
    {\rm gain}(\cS)\ \geq&\ 2\cdot \int_{t=0}^{\infty}\sum_{U\in \cU_{\cS}^t} \frac{m(U,t)}{m(U,t)+1} \  dt 
                           \ \geq\  2\cdot \big( \tfrac{1}{2} \cdot x_1 + \tfrac{1}{2} \cdot
                           \overline{x}_1 + \tfrac{2}{3} \cdot x_2 +
                           \tfrac{2}{3} \cdot \overline{x}_2 +
                           \tfrac{3}{4} \cdot x_3 \big). 
  \end{align*}
  Moreover, combining \Cref{lem:constructing_trees_from_opt_part_2} and
  (\ref{eq:bound_value}), we get
  \begin{equation}\label{eq:bound_cost_S}
    \cost(\cS) \ \leq\  x_1 + \overline{x}_1 + x_2 + \overline{x}_2 + x_3 + 2 \cdot \excess(\OPT).
  \end{equation}
  Combining these two inequalities completes the proof.
\end{proof}

Of course, this is not enough if most of the value in $\val(\cS)$
comes from $x_1 + \overline{x}_1$. In fact, the case of
$\overline{x}_1$ being large is also easy, since we can use
\Cref{lem:bad_times_yield_excess} to infer that we would have large
excess. Hence, it comes down the case where $x_1$ is large: it is
precisely in this case that we show how to distill out an autarkic
collection with large profit.

\subsubsection{Extracting the Autarkic Collection $\cP$}
\label{sec:extr-autark-coll}

\begin{definition}[$S$-relevant]
  For a set $S\in \cS$, we say that a set $U\in \supp(y)$ is
  \emph{$S$-relevant} if $\sep(U)$ contains some demand pair satisfied
  by $S$. 
\end{definition}
Observe that every $S$-relevant set $U$ separates the set $S$.

\begin{definition}[Connection Time]
  For a demand pair, let the \emph{connection time} of the pair be the
  first time $t$ where the two vertices no longer belong to different
  sets in $\cA^t$. 
\end{definition}
Note that the connection time of the pair is never more than the
deactivation time of its vertices.

\begin{lemma}\label{lem:simple_aut_pairs}
  There is a crossing-free autarkic collection $\cP$ such that
  \begin{align*}
    \coverage(\cP) \ \geq&\  x_1  \\
    \cost(\cP) \ \leq&\ x_1 + \overline{x}_1 + x_2 + \overline{x}_2 + x_3 + 2\cdot \excess(\OPT).
  \end{align*}
\end{lemma}

\begin{proof}
  To construct $\cP$, we proceed as follows, starting with
  $\cP=\emptyset$.  For each set $S\in \cS$, we consider the smallest
  $S$-good time $t'$, where $\cA^{t'}$ contains at most two
  $S$-relevant sets.  If $\cA^{t'}$ contains some $S$-relevant set,
  then there are exactly two $S$-relevant sets $A_1,A_2\in \cA^{t'}$.
  Because the sets $A_i$ with $i\in \{1,2\}$ are $S$-relevant, we have
  $\sep(A_i)\neq \emptyset$.  In order to show that $\{A_1,A_2\}$ is
  an autarkic pair, it remains to show $\sep(A_1)=\sep(A_2)$.  Since
  $\cS$ is agreeable and $S\in \cS$ satisfies some demand pair in
  $\sep(A_i)$, every demand pair in $\sep(A_1)$ is satisfied by the
  collection $\cS$.  The time $t'$ being $S$-good implies that every
  demand pair in $\sep(A_1)$ is satisfied by the set $S$.  Any demand
  pair in $\sep(A_1)$ has a connection time strictly larger than $t'$
  because $A_1\in \cA^{t'}$ and thus both vertices of such a demand
  pair must be contained in sets from $\cA^{t'}$.  But $A_1$ and $A_2$
  are the only $S$-relevant sets in $\cA^{t'}$, so one vertex of the
  demand pair must be in $A_1$ and the other in $A_2$.  This shows
  that any demand pair in $\sep(A_1)$ is also contained in
  $\sep(A_2)$.  By symmetry, we also have
  $\sep(A_2)\subseteq \sep(A_1)$.  We conclude that
  $\sep(A_1)=\sep(A_2)$ and thus $\{A_1,A_2\}$ is indeed an autarkic
  pair, which we add to~$\cP$.

  Next, we show that the autarkic collection $\cP$ is crossing-free.
  In the construction of $\cP$, we included at most one autarkic tuple
  $P_S$ for each set $S\in \cS$.  We saw above that every set $A_i$
  contained in this autarkic tuple constructed for $S\in \cS$ is
  satisfied by $S$.  Thus, every demand pair in $\sep(P_S)$ is
  satisfied by $S$.  The collection $\cS$ being agreeable implies that
  sets in $\cS$ are disjoint, so no demand pair is satisfied by two
  different sets from $\cS$.  This implies that the sets $\sep(P_S)$
  of demand pairs are disjoint, and  $\cP$ is crossing-free.

  We now bound the coverage and cost of the collection $\cP$. Since
  the tree $T_S$ connects every demand pair satisfied by any
  $S \in \cS$, the autarkic pair $P_S$ corresponding to $S$ added to
  $\cP$ also satisfies $\cost(P_S) \leq c(T_S)$. Hence using
  \Cref{lem:constructing_trees_from_opt_part_2}, we have
  \[
    \cost(\cP) \ \leq\ \val(\cS) + 2\cdot \excess(\OPT) \ = \ x_1 +
    \overline{x}_1 + x_2 + \overline{x}_2 + x_3 + 2\cdot
    \excess(\OPT). 
  \]

  Finally, it remains to prove $\coverage(\cP) \geq x_1$.  To this
  end, we fix a set $S\in \cS$ and consider the autarkic pair
  $P_S=\{A_1, A_2\} \in \cP$.  It suffices to show that if
  $U\in \cU_{S}^{t^*}$ for some $S$-good time $t^*$ with
  $m(U,t^*) = 1$, then the autarkic pair $P_S$ covers $U$.  Because
  $m(U,t^*)=1$, there are at most two $S$-relevant sets in
  $\cA^{t^*}$; here we used the fact that $\cS$ is agreeable to infer
  that every $S$-relevant set is contained in $\cU_{\cS}$. Recall that
  this pair $\{A_1, A_2\}$ were added to $\cP$ because
  $A_1, A_2 \in \cA^{t'}$, where $t'$ was the smallest $S$-good time
  where $\cA^{t'}$ contained at most two $S$-relevant sets. The
  minimal choice of $t'$ implies $t' \leq t^*$.  Using that the
  connection time of the demand pair $\varphi(U)=\{v_1,v_2\}$ is greater
  than $t^*\geq t'$, each of the two vertices $v_1$ and $v_2$ of
  $\varphi(U)$ is contained in a different set of $\cA^{t'}$.  Hence, we
  may assume $v_1 \in A_1$ and $v_2 \in A_2$ with $v_1\in U$ and
  $v_2 \notin U$.  Using that $\supp(y)$ is laminar, this implies
  $A_1 \subseteq U$ and $A_2\subseteq V\setminus U$ (because
  $v_1\in A_1 \cap U$, $v_2\in A_2 \setminus U$ and because $U$ cannot
  be a subset of any $A_i$ because of $A_i \in \cA^{t'}$,
  $U\in \cA^{t^*}$ with $t' \leq t^*$).

  Because $P_S$ is an autarkic pair, we have
  $\sep (A_1) = \sep (A_2)$.  To prove that $U$ is covered by the
  autarkic pair $P_S$, we have to show $\sep(U) = \sep (A_1)$.  Consider
  a demand pair $\{u,w\}$ in $\sep(A_1)$, where we may assume
  $u\in A_1 \subseteq U$, implying $w\in A_2 \subseteq V\setminus U$.
  Then $\{u,w\} \in \sep(U)$.  Now consider a demand pair $\{u,w\}$ in
  $\sep(U)$.  Because $\cS$ is agreeable and $S\in \cS$ satisfies the
  demand pair $\varphi(U)$, the collection $\cS$ must also satisfy
  $\{u,w\}$.  Moreover, because the time $t^*$ is $S$-good, the set
  $S\in \cS$ must satisfy $\{u,w\}$, implying $\{u,w\} \in \sep(A_1)$
  (where we use that $A_1$ and $A_2$ are the only $S$-relevant sets in
  $\cA^{t'}$ and the connection time of $\{u,w\}$ is at least
  $ t^* \geq t'$).
\end{proof}

\subsubsection{Better Autarkic Collections $\cP$}
\label{sec:bett-autark-coll}

Combining
\Cref{lem:bad_times_yield_excess,lem:bound_gain_cost,lem:simple_aut_pairs}
we can already get a constant below 2. However, we can improve the
bound in two ways: (a) we get a more refined cost bound for
collections of autarkic pairs, and (b)~we can also use autarkic
\emph{triples}. We capture these in the following two lemmas.

\begin{restatable}{lemma}{AutBetterCost}\label{lem:better_cost_bound_aut_pairs}
  There is a crossing-free autarkic collection $\cP$ such that
  \begin{align*}
    \coverage(\cP) \ \geq&\  x_1  \\
    \cost(\cP) \ \leq&\ x_1 + \tfrac{2}{3} \cdot x_2 + \tfrac{1}{2} \cdot x_3 + 6 \cdot \excess(\OPT).
  \end{align*}
\end{restatable}

\begin{restatable}{lemma}{AutarkicTriples}\label{lem:existence_autarkic_triples}
  There is a crossing-free autarkic collection $\cP$ such that
  \begin{align*}
    \coverage(\cP) \ \geq&\  x_1 + x_2 \\
    \cost(\cP) \ \leq&\ x_1 + \overline{x}_1 + x_2 + \overline{x}_2 + x_3 + 2\cdot \excess(\OPT).
  \end{align*}
\end{restatable}

We prove
\Cref{lem:better_cost_bound_aut_pairs,lem:existence_autarkic_triples}
in \Cref{sec:improving_constant}. %
In order to prove \Cref{lem:win-win}, we first 
combine them as follows:
\begin{corollary}\label{cor:existence_autarkic}
  \begin{equation*}
    2\coverage(\cP) - \cost(\cP) \ \geq x_1 - \tfrac{1}{9} \cdot x_2
    -
    \tfrac{2}{3}\cdot  x_3
    - \tfrac{1}{3} \cdot (\overline{x}_1 + \overline{x}_2)
    - \tfrac{14}{3} \cdot \excess(\OPT). 
  \end{equation*}
\end{corollary}

\begin{proof}
  Let $\cP_1, \cP_2$ be the autarkic collections obtained from
  \Cref{lem:better_cost_bound_aut_pairs,lem:existence_autarkic_triples}
  respectively, and let $\cP \in \{\cP_1, \cP_2\}$ be the
  one with higher profit $2\coverage(\cP) - \cost(\cP)$. Using that 
  $\profit(\cP) \geq \tfrac{2}{3} \profit(\cP_1) + \tfrac{1}{3}
  \profit(\cP_2)$ and substituting the numbers from above completes
  the proof.
\end{proof}

\begin{proof}[Proof of \Cref{lem:win-win}]
  Let $\cS$ be the collection of vertex sets constructed in
  \Cref{sec:definition-of-cS}, and let $\cP$ be the crossing-free
  autarkic collection resulting from \Cref{cor:existence_autarkic}.
  Then combining \Cref{lem:bound_gain_cost} and
  \Cref{cor:existence_autarkic} yields:
  \begin{align*}
    \tfrac{10}{3}\cdot (\gain(\cS) - \cost(\cS))
    &+ (2\coverage(\cP) - \cost(\cP)) +  17 \cdot \excess(\OPT) \\[2mm]
    \geq&\   \tfrac{10}{3}\cdot \big(  \tfrac{1}{3} \cdot (x_2 + \overline{x}_2) + \tfrac{1}{2} \cdot x_3 - 2\cdot \excess(\OPT) \big) \\
    &\ +  x_1 - \tfrac{1}{9} \cdot x_2 - \tfrac{2}{3}\cdot  x_3 - \tfrac{1}{3} \cdot  (\overline{x}_1 + \overline{x}_2) - \tfrac{14}{3}\cdot \excess(\OPT) \\
    &\ + 17 \cdot \excess(\OPT) \\[2mm]
    \geq&\ x_1 + x_2 + x_3 -  \tfrac{1}{3}  \cdot  \overline{x}_1 + \tfrac{7}{9} \cdot  \overline{x}_2  + \tfrac{17}{3}  \cdot \excess(\OPT) \\[2mm]
    =&\ \val(\cS) - \tfrac{4}{3}  \cdot  \overline{x}_1 - \tfrac{2}{9} \cdot  \overline{x}_2  + \tfrac{17}{3} \cdot \excess(\OPT),
  \end{align*}
  where we used \eqref{eq:bound_value} in the last inequality.

  It now suffices to show
  $\tfrac{4}{3} \cdot \overline{x}_1 + \tfrac{2}{9} \cdot
  \overline{x}_2 \leq \tfrac{17}{3} \cdot
  \excess(\OPT)$. 
  For a set $S\in \cS$ and a time $t\geq 0$, we have
  $|\cU_S^t| \leq m(U,t) +1$ for all $U\in \cU_S^t$ and thus
    \begin{align*}
      \tfrac{1}{2} \cdot   |\{ U\in \cU_{S}^t: m(U,t) = 1\}|  +  \tfrac{1}{3} \cdot   |\{ U\in \cU_{S}^t: m(U,t) = 2\}|    \ \leq\ 1.
    \end{align*}
  Using this together with the definitions of $\overline{x}_1$, $\overline{x}_2$,
  and $B(\cS)$ (from \Cref{lem:bad_times_yield_excess}) yields
  \begin{align*}
    \tfrac{1}{2} \overline{x}_1 + \tfrac{1}{3} \overline{x}_2 \ =
    &\  \int_{t=0}^\infty \sum_{\substack{S\in \cS: \\ t\text{ is $S$-bad}}}  \Big(\tfrac{1}{2} \cdot   |\{ U\in \cU_{S}^t: m(U,t) = 1\}|  +  \tfrac{1}{3} \cdot   |\{ U\in \cU_{S}^t: m(U,t) = 2\}| \Big)  \ dt. \\
  \leq&\   \int_{t=0}^{\infty}  |\{ S \in \cS : t \text{ is $S$-bad}\}| \ dt 
  \ =\ B(\cS).
  \end{align*}
   By
  \Cref{lem:bad_times_yield_excess}, this implies
  \[
    \tfrac{4}{3} \cdot \overline{x}_1 + \tfrac{2}{9} \cdot
    \overline{x}_2\ \leq\ \tfrac{8}{3} \cdot \big( \tfrac{1}{2}
    \overline{x}_1 + \tfrac{1}{3} \overline{x}_2 \big) \ \leq\
    \tfrac{8}{3}\cdot B(\cS) \ \leq\ \tfrac{16}{3}
    \cdot \excess(\OPT)),
  \]
  which suffices to complete the proof.
\end{proof}

\section{The Approximation Ratio}
\label{sec:approximation-ratio}

Now that we have all the pieces in hand, we can show an approximation
ratio of $\appx$. For ready reference, all the relevant bounds are
given in \Cref{fig:constraints}.

\begin{figure}[ht]
  \centering
  \lineshere
  \begin{alignat*}{2}
    c(F_1) &\leq  2(1+\epsilon)\cdot y(\Usep)  & &
                                                   \text{(\Cref{thm:AKR-GW})} \\
    c(F_2) &\leq  2(1+\eps)\cdot y(\Usep) + (\alpha+\delta)\, \cost(\cS) -
             (1-e^{-\alpha})\, \gain(\cS) & \qquad & \text{(\Cref{thm:effic-contraction})}\\
    c(F_3) &\leq  2\, c(\OPT) - \profit(\cP) + 2\excess(\OPT) && \text{(\Cref{cor:autarkic_pairs})}\\
    y(\Usep) &\leq c(\OPT) - \excess(\OPT) && \text{(\Cref{eq:excess})}\\
    \cost(\cS) &\leq \val(\cS) + 2 \cdot \excess(\OPT) && \text{(\Cref{lem:constructing_trees_from_opt_part_2})}\\
    y(\Usep) &\leq \val(\cS) +  \nicefrac{1}{2\epsilon} \,\excess(\OPT)
                                               &&
                                                  \text{(\Cref{lem:constructing_trees_from_opt_part_3})} \\
    \val(\cS) &\leq \nicefrac{10}{3}\cdot (\gain(\cS) - \cost(\cS)) +
                \profit(\cP) +
                17 \cdot \excess(\OPT) &&
                                                           \text{(\Cref{lem:win-win})} 
  \end{alignat*}
    \lineshere
  \caption{Summary of Bounds}
  \label{fig:constraints}
\end{figure}

With the benefit of hindsight, we set $\eps = 0.0083$ and
$\alpha + \delta = 0.1$, and choose $\delta > 0$ to be a small enough
constant, so that $(1-e^{-\alpha}) \geq 0.095$. Note this value of
$\alpha$ may not be optimal, but it still gives a valid bound.

\begin{enumerate}
\item If $\excess(\OPT) \geq 0.0116 \, c(\OPT)$, then
  using \Cref{thm:AKR-GW} (or \Cref{thm:effic-contraction} with $\alpha =
  0$ and $\delta$ small enough) and the definition of excess implies that
  \[ c(F_1) \leq  2(1+\epsilon)\cdot (1 - 0.0116)\, c(\OPT) \leq
    1.994\, c(\OPT). \]
  We can henceforth assume that $\excess(\OPT) \leq 0.0116 \, c(\OPT)$.
\item If $\profit(\cP) \geq 0.03\, c(\OPT)$, then
  \Cref{cor:autarkic_pairs} implies that
  \[ c(F_3) \leq c(\OPT) \cdot [ 2 - 0.03 + 2 \cdot 0.0116 ] \leq
    \appx \, c(\OPT). \]
  Henceforth, we  assume that $\profit(\cP) \leq
  0.03\, c(\OPT)$, which we substitute into the win-win bound from \Cref{lem:win-win}.

\item At this point, we have the following five constraints:
  \begin{align*}
    c(F_2) \ &\leq  2(1+\eps)\cdot y(\Usep) + 0.1\, \cost(\cS) - 0.095\, \gain(\cS)\\
    y(\Usep) \  &\leq c(\OPT) - \excess(\OPT) \\
    \cost(\cS) \ &\leq \val(\cS) + 2 \cdot \excess(\OPT) \\
    \val(\cS) \ &\leq \nicefrac{10}{3}\cdot (\gain(\cS) - \cost(\cS)) +
                0.03 \, c(\OPT) +
                17 \cdot \excess(\OPT) \\
    y(\Usep) \ &\leq \val(\cS) +  \nicefrac{1}{2\epsilon} \,\excess(\OPT),
  \end{align*}
  where $\eps = 0.0083$.
  If we multiply these constraints by $1, 1.9931, 0.005, 0.0285$ and
  $0.0235$ respectively, and sum them up, we get
  \begin{gather}
    c(F_2) +  0.0829 \cdot \excess(\OPT) \leq 1.994 \cdot c(\OPT).
  \end{gather}
  Since the excess is non-negative, this completes the proof of the approximation.
\end{enumerate}

 \subsection*{Acknowledgments}
 We thank Chandra Chekuri
 for helpful conversations.

{\small
\bibliographystyle{alpha}
\bibliography{refs}

\newcommand{\etalchar}[1]{$^{#1}$}
\begin{thebibliography}{AGH{\etalchar{+}}25b}

\bibitem[AGH{\etalchar{+}}24]{AhmadiGHJM24}
Ali Ahmadi, Iman Gholami, MohammadTaghi Hajiaghayi, Peyman Jabbarzade, and
  Mohammad Mahdavi.
\newblock Prize-collecting {Steiner} tree: A 1.79 approximation.
\newblock In {\em STOC}, page 1641–1652, 2024.

\bibitem[AGH{\etalchar{+}}25a]{AhmadiGHJM25}
Ali Ahmadi, Iman Gholami, MohammadTaghi Hajiaghayi, Peyman Jabbarzade, and
  Mohammad Mahdavi.
\newblock 2-approximation for prize-collecting {Steiner} forest.
\newblock {\em J. {ACM}}, 72(2):17:1--17:27, 2025.

\bibitem[AGH{\etalchar{+}}25b]{AGHJM25}
Ali Ahmadi, Iman Gholami, MohammadTaghi Hajiaghayi, Peyman Jabbarzade, and
  Mohammad Mahdavi.
\newblock Breaking a long-standing barrier: {$2-\varepsilon$} approximation for
  {Steiner} forest.
\newblock {\em CoRR}, abs/2504.11398, 2025.

\bibitem[AKR95]{AKR95}
Ajit Agrawal, Philip Klein, and R.~Ravi.
\newblock When trees collide: an approximation algorithm for the generalized
  {Steiner} problem on networks.
\newblock {\em SIAM J. Comput.}, 24(3):440--456, 1995.

\bibitem[BD97]{BorchersD97}
Al~Borchers and Ding{-}Zhu Du.
\newblock The {$k$-Steiner} ratio in graphs.
\newblock {\em {SIAM} J. Comput.}, 26(3):857--869, 1997.

\bibitem[BGRS13]{ByrkaGRS13}
Jaros{\l}aw Byrka, Fabrizio Grandoni, Thomas Rothvoss, and Laura Sanit{\`a}.
\newblock Steiner tree approximation via iterative randomized rounding.
\newblock {\em J. ACM}, 60(1):Art. 6, 33, 2013.

\bibitem[BGT24]{Byrka0T24}
Jaroslaw Byrka, Fabrizio Grandoni, and Vera Traub.
\newblock The bidirected cut relaxation for {Steiner} tree has integrality gap
  smaller than 2.
\newblock In {\em {FOCS} 2024}, pages 730--753. {IEEE}, 2024.

\bibitem[BGT25]{ByrkaGT25}
Jaroslaw Byrka, Fabrizio Grandoni, and Vera Traub.
\newblock On the bidirected cut relaxation for {Steiner} forest.
\newblock In {\em Integer Programming and Combinatorial Optimization, {IPCO}
  2025,}, pages 114--127. Springer, 2025.

\bibitem[CC08]{ChlebikC08}
Miroslav Chleb{\'{\i}}k and Janka Chleb{\'{\i}}kov{\'{a}}.
\newblock The {Steiner} tree problem on graphs: Inapproximability results.
\newblock {\em Theor. Comput. Sci.}, 406(3):207--214, 2008.

\bibitem[CV25]{CV25}
Chandra Chekuri and Jan Vondr\'ak.
\newblock {\em Maximizing Submodular Set Functions}.
\newblock NOW Publishing, 2025.
\newblock to appear.

\bibitem[DW71]{DreyfusW71}
Stuart~E. Dreyfus and Robert~A. Wagner.
\newblock The {Steiner} problem in graphs.
\newblock {\em Networks}, 1(3):195--207, 1971.

\bibitem[GGK{\etalchar{+}}18]{GGKMSSV18}
Martin Gro{\ss}, Anupam Gupta, Amit Kumar, Jannik Matuschke, Daniel~R. Schmidt,
  Melanie Schmidt, and Jos{\'{e}} Verschae.
\newblock A local-search algorithm for {Steiner} forest.
\newblock In {\em ITCS}, pages 31:1--31:17, Jan 2018.

\bibitem[GK15]{GK15-stoc}
Anupam Gupta and Amit Kumar.
\newblock Greedy algorithms for {Steiner} forest.
\newblock In {\em STOC}, Jun 2015.

\bibitem[GKPR03]{GuptaKPR03}
Anupam Gupta, Amit Kumar, Martin P{\'{a}}l, and Tim Roughgarden.
\newblock Approximation via cost-sharing: {A} simple approximation algorithm
  for the multicommodity rent-or-buy problem.
\newblock In {\em FOCS}, pages 606--615. {IEEE} Computer Society, 2003.

\bibitem[GORZ12]{GoemansORZ12}
Michel~X. Goemans, Neil Olver, Thomas Rothvo{\ss}, and Rico Zenklusen.
\newblock Matroids and integrality gaps for hypergraphic {Steiner} tree
  relaxations.
\newblock In {\em STOC}, pages 1161--1176. {ACM}, 2012.

\bibitem[GW95]{GW95}
Michel~X. Goemans and David~P. Williamson.
\newblock A general approximation technique for constrained forest problems.
\newblock {\em SIAM J. Comput.}, 24(2):296--317, 1995.

\bibitem[HJ06]{HajiaghayiJ06}
Mohammad~Taghi Hajiaghayi and Kamal Jain.
\newblock The prize-collecting generalized {Steiner} tree problem via a new
  approach of primal-dual schema.
\newblock In {\em {SODA}}, pages 631--640, 2006.

\bibitem[Jai01]{Jain98}
Kamal Jain.
\newblock A factor 2 approximation algorithm for the generalized {Steiner}
  network problem.
\newblock {\em Combinatorica}, 21(1):39--60, 2001.
\newblock (Preliminary version in {\em 39th FOCS}, pages 448--457, 1998).

\bibitem[KLS05]{KLS05}
Jochen K{\"o}nemann, Stefano Leonardi, and Guido Sch{\"a}fer.
\newblock A group-strategyproof mechanism for {Steiner} forests.
\newblock In {\em SODA}, pages 612--619, 2005.

\bibitem[KLSvZ08]{KLSZ08}
Jochen K{\"o}nemann, Stefano Leonardi, Guido Sch{\"a}fer, and Stefan H.~M. van
  Zwam.
\newblock A group-strategyproof cost sharing mechanism for the {Steiner} forest
  game.
\newblock {\em SIAM J. Comput.}, 37(5):1319--1341, 2008.

\bibitem[RZ05]{RobinsZ05}
Gabriel Robins and Alexander Zelikovsky.
\newblock Tighter bounds for graph {S}teiner tree approximation.
\newblock {\em SIAM J. Discrete Math.}, 19(1):122--134, 2005.

\bibitem[TZ25]{TraubZ25}
Vera Traub and Rico Zenklusen.
\newblock Better-than-2 approximations for weighted tree augmentation and
  applications to {Steiner} tree.
\newblock {\em J. {ACM}}, 72(2):16:1--16:40, 2025.

\bibitem[Zel96]{zelikovsky_1996_better}
A.~Zelikovsky.
\newblock Better approximation bounds for the network and {Euclidean} {Steiner}
  tree problems.
\newblock Technical report, University of Virginia, 1996.
\newblock CS-96-06.

\end{thebibliography}
}

\newpage

\appendix

\section{Example: Criticality of Active Connectedness}
\label{sec:crit-active-conn}

Consider the example in \Cref{fig:actively_connected}.

\begin{figure}[h]
\begin{center}
  \begin{tikzpicture}[scale=0.5]

\tikzset{vertex/.style={
fill=black,thick, circle,minimum size=5pt, inner sep=0pt, outer sep=1.5pt}
}

\begin{scope}[every node/.style={vertex}]
\node (a1) at (0,0) {};
\node (a2) at (3,0) {};
\node (a3) at (8,0) {};
\node (b2) at (5,0) {};
\node (b1) at (19,0) {};
\node (b3) at (20,0) {};
\end{scope}

\node[below=1pt] at (a1) {$a_1$};
\node[below=1pt] at (b1) {$b_1$};
\node[below=1pt] at (a2) {$a_2$};
\node[below=1pt] at (b2) {$b_2$};
\node[below=1pt] at (a3) {$a_3$};
\node[below=1pt] at (b3) {$b_3$};

\begin{scope}[very thick]
\draw (a1)  -- node[above] {3} (a2);
\draw (a2)  -- node[above] {2} (b2);
\draw (b2)  -- node[above] {3} (a3);
\draw (a3)  -- node[above] {$M$} (b1);
\draw (b1)  -- node[above] {1} (b3);
\end{scope}

\end{tikzpicture}
\end{center}
\vspace*{-5mm}
\caption{\label{fig:actively_connected}
An instance of Steiner Forest with demand pairs $\{a_1,b_1\}, \{a_2,b_2\}, \{a_3,b_3\}$.
}
\end{figure}
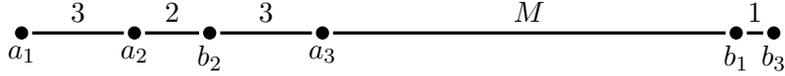

At time $t^* = 2$, the vertices $S = \{a_1, a_2, b_2, a_3\}$ form an
active component. But $S$ is not an actively connected set, since the
vertices $\{a_2,b_2\}$ were deactivated at time $1+2\eps$. Now if we
were to apply the result from \cite{BorchersD97} to form 
$3$-restricted Steiner trees for $S$, the resulting
collection could be $\cS_3 = \{\{a_1, a_2, b_2\}, \{b_2,
a_3\}\}$. Assume that $\eps = 0.1$, say: then at time $t = 1.5$, the
sets in $\cS_3$ do not merge any sets in $\cA^t$, whereas the
(not-actively-connected) set $S$ merges the active components
$\{a_1\}, \{a_3\}$. This would mean that $\cA^t/\cS \neq \cA^t/\{S\}$.

\section{Submodular Optimization with Knapsacks}
\label{sec:submod}

We are given a submodular function $f: 2^X \to \RR_{\geq 0}$, where
$f(\emptyset) = 0$, and each element $x \in X$ has a cost $c(x)$. For
simplicity, assume that the costs are distinct. Let
$S^* = \arg\max_{S \sse X} [f(S) - c(S)]$. We claim the following:
\begin{theorem}
  There exists an algorithm that given $\alpha \in [0,1]$ and
  $\gamma > 0$, finds a set $S \sse X$ such that
  \begin{gather}
    f(S) - c(S) \geq \big(1 - e^{-\alpha}\big) \cdot f(S^*) - (\alpha+\gamma)
    \cdot c(S^*). \label{eq:5}
  \end{gather}
  Moreover, the algorithm runs in time $n^{O(1/\gamma)}$.
\end{theorem}

\begin{proof}
  To prove this, let us first consider the case where each element has
  cost $c(e) \leq C_{\max}$ for some value $C_{\max}$. In this case,
  consider the algorithm that starts with the empty set and picks
  elements greedily, setting
  $e_{i+1} = \arg\max_{e \in X} \frac{f_{S_{i}}(e)}{c(e)}$, where
  $S_i = \{e_1, \ldots, e_i\}$ is the set of the first $i$ elements to
  be picked and $f_{S_{i}}(e) \coloneqq f(S_{i} \cup \{e\}) - f(S_i)$. 
  Return the set $S_i$ that maximizes $f(S_i) - c(S_i)$. It
  suffices to show that there exists a good set.

  Indeed, consider the first index $k$ such that
  $c(S_k) \geq \alpha c(S^*)$; a
  standard submodularity-based analysis shows that:
  \[ f(S_k) \geq \big( 1 - e^{-c(S_k)/c(S^*)} \big) \cdot f(S^*)\geq
    (1-e^{-\alpha}) \cdot f(S^*). \] (See, e.g.,
  \cite[Proposition~2.2]{CV25}; for completeness we present their proof
  in \Cref{prop:greedy}.) Since the elements have cost at most $C_{\max}$, we know
  that $c(S_k) \leq \alpha \, c(S^*) + C_{\max}$, we get that
  \[ f(S_k) - c(S_k) \geq (1-e^{-\alpha}) \cdot f(S^*) - (\alpha \,
    c(S^*) + C_{\max}). \] If $C_{\max} \leq \gamma c(S^*)$, we would
  get the bound claimed in (\ref{eq:5}).

  To handle the general case where the costs may be larger than
  $\gamma \, c(S^*)$, we need a partial enumeration strategy. We guess
  the (at most) $1/\gamma$ elements in $S^*$ that have cost more than
  $\gamma \, c(S^*)$. Moreover, since we do not know the cost
  $c(S^*)$, we also guess the maximum-cost element in $X$ that is at
  most $\gamma \, c(S^*)$, so that we can define the set of
  ``low-cost'' elements on which to run the above algorithm. Indeed,
  the algorithm for the general case is the following:
  \begin{enumerate}
  \item Enumerate over sets $S'$ of at most $1/\gamma$ elements, and
    one extra element $e_0$.
  \item Define $X' := \{e \in X \mid c(e) \leq c(e_0)\}$, and define
    submodular function $g(S) := f_{S'}(S)$ over the ground set
    $X'$. Use the above procedure on submodular function $g$ over
    ground set $X'$ to find a solution $S''$.
  \item Among all these solutions, return the solution with largest
    $f(S) - c(S)$, where $S = S' \cup S''$.
  \end{enumerate}
  To prove the approximation guarantee, consider the run where we
  correctly guess the elements $S'$ to be those in $S^*$ having cost
  strictly more than $\gamma\, c(S^*)$, and also correctly guess $e_0$
  so that $X' = \{e \in X \mid c(e) \leq \gamma \, c(S^*)\}$. Define
  $H^* = S^* \setminus X'$ and $L^* = S^* \cap X'$; by assumption, we
  are focusing on the guess where $H^* = S'$, and also $L^* \sse X'$
  with $C_{\max} \leq \gamma\, c(S^*)$. Now the algorithm's
  solution $S = S' \cup S''$ for this run satisfies
  \begin{align*}
    f(S) - c(S)
    &= f(S') - c(S') + f_{S'}(S'') - c(S'') \\
    &\geq f(H^*) - c(H^*) + (1-e^{-\alpha}) \cdot f_{H^*}(L^*) -
      (\alpha \cdot c(L^*) + \gamma \cdot c(S^*)).
      \intertext{We will soon show the claim that $f(H^*) - c(H^*) \geq
      0$, but first we use it to complete the proof. Multiplying the
      first two terms by $\alpha\in[0,1]$ gives}
    &\geq \alpha(f(H^*) - c(H^*)) + (1-e^{-\alpha}) \cdot f_{H^*}(L^*) -
      (\alpha \cdot c(L^*) + \gamma \cdot c(S^*) ).
      \intertext{Moreover, using that $\alpha \geq 1 - e^{-\alpha}$, we get}
    &\geq (1-e^{-\alpha}) \cdot \big[ f(H^*) + f_{H^*}(L^*) \big] -
      \alpha \cdot \big[ c(H^*) + c(L^*) \big] - \gamma \cdot c(S^*),
  \end{align*}
  which proves the claimed bound using submodularity. Finally, to prove the claim, note that if
  $f(H^*) < c(H^*)$, we get that $f(L^*) - c(L^*) > f(S^*) - c(S^*)$,
  which contradicts the optimality of the set $S^*$.
\end{proof}

\begin{proposition}
  \label{prop:greedy}
  For the sets $S_i$ defined by the greedy process,  $f(S_k) \geq (1 -
  e^{-c(S_k)/c(S^*)}) \cdot f(S^*)$.
\end{proposition}

\begin{proof}
  By submodularity and the greedy choice of element $e_{i+1}$, it follows that
  \[ f_{S_i}(e_{i+1}) \geq \frac{c(e_{i+1})}{c(S^*)} \cdot \sum_{e \in
      S^*} f_{S_i}(e) \geq \frac{c(e_{i+1})}{c(S^*)} \cdot (f(S^*) -
    f(S_i)). \] If we define $\Delta_i = f(S^*) - f(S_i)$, the
  definition of marginal values $f_{S_i}(e_{i+1})$ implies that
  \[ \Delta_{i+1} \leq \Delta_i \cdot ( 1- c(e_{i+1})/c(S^*) ) \leq
    \Delta_i \cdot e^{-c(e_{i+1})/c(S^*)}. \] By induction, we get
  \[ \Delta_k \leq \Delta_0 \cdot e^{-c(S_k)/c(S^*)} \implies f(S_k) \geq (1 -
    e^{-c(S_k)/c(S^*)}) \cdot f(S^*),\] where we use the
  fact that $f(S_0) = 0$ and hence $\Delta_0 = f(S^*)$.
\end{proof}

\section{Better autarkic collections}
\label{sec:improving_constant}

In this section we provide the improved constructions of autarkic
collections, and prove
\Cref{lem:better_cost_bound_aut_pairs,lem:existence_autarkic_triples}.

\AutBetterCost*
\begin{proof}
  We consider the same crossing-free autarkic collection $\cP$ as in
  the proof of \Cref{lem:simple_aut_pairs}, which in particular
  satisfies $\coverage(\cP) \geq x_1$.  It remains to prove the
  claimed upper bound on $\cost(\cP)$.

  As in the proof of \Cref{lem:constructing_trees_from_opt_part_2},
  for a connected component $T$ of $\OPT$, we define
  $ \excess(T) \coloneqq c(T) - \val(\{V(T)\})$.  Then
  $ \excess(T) \geq 0$ and
  \[
    \excess(\OPT)
    = \sum_{\substack{T:\text{ conn. comp.}\\ \text{of OPT}}} \excess(T).
  \]
  Moreover, we denote the contribution of a set $S\in \cS$ to $x_1$,
  $x_2$, and $x_3$ by $x_1(S)$, $x_2(S)$, and $x_3(S)$, respectively,
  that is
  \begin{align*}
    x_1(S) \ \coloneqq&\ \int_{t=0}^\infty  |\{ U\in \cU_{S} \cap \cA^t : m(U,t) = 1\}| \cdot \mathbf{1}_{[t\text{ is $S$-good}]}   \ dt \\
    x_2(S) \ \coloneqq&\ \int_{t=0}^\infty    |\{ U\in \cU_{S}  \cap \cA^t: m(U,t) = 2\}|  \cdot \mathbf{1}_{[t\text{ is $S$-good}]}   \ dt \\
    x_3(S) \ \coloneqq&\ \int_{t=0}^\infty   |\{ U\in \cU_{S}  \cap \cA^t: m(U,t) \geq 3 \}|   \ dt ,
  \end{align*}
  where $\mathbf{1}_{[t\text{ is $S$-good}]} $ denotes the indicator variable that is $1$ if time $t$ is $S$-good and $0$ otherwise.
  The contribution of $S\in \cS$ to $B(\cS)$ will be denoted by $B(S)$, that is
  \[
    B(S)  \ \coloneqq\   \int_{t=0}^{\infty}  \mathbf{1}_{[t\text{ is $S$-bad}]}\ dt.
  \]
  Recall that $\cP$ contains (at most) one autarkic pair $P$ for each
  $S\in \cS$.  Consider $S\in\cS$ with a corresponding autarkic pair
  $P=\{A_1,A_2\} \in \cP$.  Because $\cS$ is agreeable, there is some
  connected component $T$ of $\OPT$ containing all vertices from $S$
  (and these connected components of $\OPT$ are distinct for different
  sets $S\in \cS$).  Let $\{v,w\} \in \sep(A_1)$ be the designated
  representative of the autarkic pair $P$.  We will show that
  $\cost(\{P\})$, that is the length of a shortest $v$-$w$ path,
  satisfies
  \begin{equation}\label{eq:cost_bound_per_autarkic_pair}
    \cost(\{P\}) \ \leq\  x_1(S) + \tfrac{2}{3} \cdot x_2(S) + \tfrac{1}{2} \cdot x_3(S) +2 \cdot \excess(T) +2 \cdot B(S).
  \end{equation}
  This will complete the proof as it implies
  \begin{align*}
    \cost(\cP) \ =&\ \sum_{P\in \cP} \cost(\{P\})  \\
    \ \leq&\ \sum_{S\in \cS} \big(  x_1(S) + \tfrac{2}{3} \cdot x_2(S) + \tfrac{1}{2} \cdot x_3(S) + 2 \cdot B(S)\big)  
            +\sum_{\substack{T:\text{ conn. comp.}\\ \text{of OPT}}} 2 \cdot  \excess(T)\\
    =&\ x_1 + \tfrac{2}{3} \cdot x_2 + \tfrac{1}{2} \cdot x_3 + 2 \cdot  B(\cS) +  2 \cdot \excess(\OPT) \\
    \leq&\  x_1 + \tfrac{2}{3} \cdot x_2 + \tfrac{1}{2} \cdot x_3 +  6 \cdot \excess(\OPT)
  \end{align*}
  where the last inequality follows from
  \Cref{lem:bad_times_yield_excess}.

  It remains to prove \eqref{eq:cost_bound_per_autarkic_pair}.  To
  this end, we let $R$ be the $v$-$w$ path in $T$ and observe that
  $c(R)$ is an upper bound on $\cost(\{P\})$.  We have
  $c(R) =c(T) - c(T\setminus R)$.  If the tree $T$ contains only a
  single edge in a cut $\delta(U)$ with $|\{v,w\}\cap U|\neq1$, then
  this edge is contained in $T\setminus R$.  Using that $y$ satisfies
  the constraints of the dual LP, this implies
  \begin{equation*}
    c(T\setminus R) \ \geq\ y\big( \big\{ U\in \supp(y) : |T\cap
    \delta(U)|=1 \text{ and } |\{v,w\} \cap U| \neq 1 \big\}\big). 
  \end{equation*}
  Moreover, we have
  \begin{equation*}
    \sum_{U\in \supp(y)} |T\cap \delta(U)| \cdot y_U 
    \ \leq\ c(T) 
    \ =\ \val(\{V(T)\}) + \excess(T) 
    \ \leq\ \sum_{\substack{U\in \supp(y): \\|T\cap \delta(U)| \geq 1}} y_U  + \excess(T)
  \end{equation*}
  implying
  \[
    \sum_{\substack{U\in \supp(y): \\|T\cap \delta(U)| \geq 2}} y_U \ \leq\ \excess(T)
  \]
  and thus 
  \begin{equation*}
    c(T)\ \leq\ \sum_{\substack{U\in \supp(y): \\|T\cap \delta(U)| \geq 1}} y_U  + \excess(T)
    \ \leq\  \sum_{\substack{U\in \supp(y): \\|T\cap \delta(U)| = 1}} y_U + 2\cdot \excess(T).
  \end{equation*}
  We conclude
  \begin{align*}
    c(R) \ =\ c(T) - c(T\setminus R) \ \leq
    &\  \sum_{\substack{U\in \supp(y): \\|T\cap \delta(U)| = 1}} y_U + 2\cdot \excess(T) \\
    &\ - y\big( \big\{ U\in \supp(y) : |T\cap \delta(U)|=1 \text{ and
      } |\{v,w\} \cap U| \neq 1 \big\}\big) \\[2mm]
    \leq&\ y\big( \big\{ U\in \supp(y) : |\{v,w\} \cap U| = 1 \big\}\big) + 2\cdot \excess(T).
  \end{align*}
  Hence, in order to prove \eqref{eq:cost_bound_per_autarkic_pair}, it remains to prove
  \begin{equation}
    y\big( \big\{ U\in \supp(y) : |U\cap \{v,w\}| = 1 \big\}\big)  \ \leq\ x_1(S) + \tfrac{2}{3} \cdot x_2(S) + \tfrac{1}{2} \cdot x_3(S)  + 2 \cdot B(S).
  \end{equation}
  Let $t^*$ be the connection time of the demand pair $\{v,w\}$.
  Because for every time $t$, the active sets $\cA^t$ are disjoint and thus at most two of them can contain a vertex from $\{v,w\}$, we have
  \begin{align*}
    y\big( \big\{ U\in \supp(y) : |U\cap \{v,w\}| = 1 \big\}\big) 
    \ =&\ \int_{t=0}^{\infty} \big|\big\{U\in \cA^t:  |U\cap \{v,w\}| = 1\big\}\big| 
    \\
    \ \leq &\   \int_{t=0}^{t^*} 2 \cdot \mathbf{1}_{[t\text{ $S$-good}]}\ dt  + 2\cdot B(S).
  \end{align*}
  Now consider an $S$-good time $t$.
  Then, defining $k\coloneqq |\cU_S \cap \cA^t|$, we have $m(U,t) =k -1$ for every set $U\in \cU_S \cap \cA^t$ because $t$ is $S$-good.
  Thus, 
  \begin{align*}
    2\ \leq&\ |\{ U\in \cU_{S} \cap \cA^t : m(U,t) = 1\}|  \\
           &\ + \tfrac{2}{3} \cdot   |\{ U\in \cU_{S}  \cap \cA^t: m(U,t) = 2\}| \\
           &\ + \tfrac{1}{2} \cdot |\{ U\in \cU_{S}  \cap \cA^t: m(U,t) \geq 3 \}|.
  \end{align*}
  By definition of $x_1(S)$, $x_2(S)$, and $x_3(S)$, this implies 
  \[
    \int_{t=0}^{t^*} 2 \cdot \mathbf{1}_{[t\text{ $S$-good}]}\ dt \ \leq\
    x_1(S) + \tfrac{2}{3} \cdot x_2(S) + \tfrac{1}{2} \cdot x_3(S). \qedhere
  \]
\end{proof}

\AutarkicTriples*

\begin{proof}
  To construct $\cP$, we proceed similarly as in the proof of
  \Cref{lem:simple_aut_pairs}, but we will now make use of autarkic
  \emph{triples}.  We start with $\cP=\emptyset$.  Then, for each
  $S\in \cS$, we consider the smallest $S$-good time $t$ such that
  $\cA^t$ contains at most three $S$-relevant sets.  If $\cA^t$
  contains some $S$-relevant set, then there are either exactly two or
  exactly three $S$-relevant sets $A_i$.  If there are exactly two
  such sets $A_i$, then $\{A_1, A_2\}$ is an autarkic pair (by the
  same argument as in the proof of \Cref{lem:simple_aut_pairs}) and we
  add this autarkic pair to $\cP$.  If there are exactly three such
  sets $A_i$, then we will show that $\{A_1,A_2,A_3\}$ is an autarkic
  triple and we will add it to $\cP$.

  To see that $\{A_1,A_2,A_3\}$ is indeed an
  autarkic triple, we need to show
  $\sep(A_1 \cup A_2 \cup A_3) =\emptyset$.  Because
  $\sep(A_i)\neq \emptyset$ for all $i\in \{1,2,3\}$ and $\cS$ is
  agreeable, every demand pair in $\sep(A_i)$ is satisfied by the
  collection $\cS$.  Using that the time $t$ is $S$-good, this
  implies that every demand pair in $\sep(A_i)$ is satisfied by the
  set $S$.  Because $A_1, A_2, A_3$ are the only $S$-relevant sets in
  $\cA^{t}$, every demand pair in $\sep(A_i)$ has one vertex in $A_i$
   and the other vertex in $A_j$ for some $j\neq i$.
   In particular, such a demand pair in $\sep(A_i)$ is not contained in
  $\sep(A_1 \cup A_2 \cup A_3)$. 
   Using $\sep(A_1 \cup A_2 \cup A_3)\subseteq \sep(A_1) \cup \sep(A_2) \cup
  \sep(A_3)$, this implies $\sep(A_1 \cup A_2 \cup A_3) =\emptyset$.
  Hence, $\{A_1, A_2, A_3\}$ is indeed an autarkic triple.

  The autarkic collection $\cP$ is crossing-free for the same reason
  as in the proof of \Cref{lem:simple_aut_pairs}: In the construction
  of $\cP$, we included at most one autarkic tuple $P_S$ for each set
  $S\in \cS$.  We have shown that every set $A_i$ contained in this
  autarkic tuple constructed for $S\in \cS$ is satisfied by $S$.
  Thus, every demand pair in $\sep(P_S)$ is satisfied by $S$.  Because
  the sets in $\cS$ are disjoint (because $\cS$ is agreeable), no
  demand pair is satisfied by two different sets from $\cS$.  This
  implies that the sets $\sep(P_S)$ of demand pairs are disjoint.
  Hence, $\cP$ is crossing-free.

  Because for every set $S\in \cS$, the tree $T_S$ connects every
  demand pair satisfied by $S$, the autarkic tuple $P_S$ added to
  $\cP$ for $S$ satisfies $\cost(P_S) \leq c(T_S)$ and thus by
  \Cref{lem:constructing_trees_from_opt_part_2}, we have
  \[
    \cost(\cP) \ \leq\ \val(\cS) + 2\cdot \excess(\OPT) \ = \ x_1 + \overline{x}_1 + x_2 + \overline{x}_2 + x_3 + 2\cdot \excess(\OPT).
  \]
  Finally, it remains to prove $\coverage(\cP) \geq x_1 + x_2$.  To
  this end, we proceed analogously as in the proof of
  \Cref{lem:simple_aut_pairs}: We fix a set $S\in \cS$ and consider
  the autarkic tuple $P_S \in \cP$.  It suffices to show that if
  $U\in \cU_{S} \cap \cA^{t^*}$ for some $S$-good time $t^*$ with
  $m(U,t^*) \in \{1,2\}$, then $P_S$ covers $U$.  If $P_S$ is an
  autarkic pair, this follows from the proof of
  \Cref{lem:simple_aut_pairs}.  Thus, we may assume that $P_S$ is an
  autarkic triple $P_S=\{A_1,A_2,A_3\}$.

  Let $U\in \cU_{S} \cap \cA^{t^*}$ for some $S$-good time $t^*$ with
  $m(U,t^*) \in \{1,2\}$.  Because $m(U,t^*)\leq 2$, there are at most
  three $S$-relevant sets in $\cA^{t^*}$, where we used that every
  $S$-relevant set is contained in $\cU_{\cS}$ because $\cS$ is
  agreeable.  As in the construction of $\cP$, we let $t'$ be the
  smallest $S$-good time, where $\cA^{t'}$ contains at most three
  $S$-relevant sets.  Then we have $A_i \in \cA^{t'}$ for all
  $i\in \{1,2,3\}$.  The minimal choice of $t'$ implies $t' \leq t^*$.

  Using that the connection time of the demand pair
  $\varphi(U)=\{v_1,v_2\}$ is greater than $t^*\geq t'$, each of the two
  vertices $v_1$ and $v_2$ of $\varphi(U)$ is contained in a different
  set of $\cA^{t'}$.  Hence, we may assume $v_1 \in A_1$ and
  $v_2 \in A_2$ with $v_1\in U$ and $v_2 \notin U$.  Using that
  $\supp(y)$ is laminar, this implies $A_1 \subseteq U$ and
  $A_2\subseteq V\setminus U$ (because $v_1\in A_1 \cap U$,
  $v_2\in A_2 \setminus U$ and because $U$ cannot be a subset of any
  $A_i$ because of $A_i \in \cA^{t'}$, $U\in \cA^{t^*}$ with
  $t' \leq t^*$).  Moreover, we either have $A_3 \cap U = \emptyset$
  or $A_3 \subseteq U$.  We distinguish these two cases and show that
  $U$ is covered by $P_S$ in each of them.

  First consider the case $A_3 \cap U = \emptyset$. We claim
  $\sep(U) = \sep (A_1)$.  To prove this, consider this a demand pair
  $\{u,w\}$ in $\sep(A_1)$, where we may assume
  $u \in A_1 \subseteq U$, which by the definition of autarkic triples
  implies $w \in A_2 \cup A_3 \subseteq V\setminus U$.  Then
  $\{u,w\} \in \sep(U)$.  Conversely, consider a demand pair $\{u,w\}$
  in $\sep(U)$.  Because $\cS$ is agreeable and $S\in \cS$ satisfies
  the demand pair $\varphi(U)$, the collection $\cS$ must also satisfy
  $\{u,w\}$.  Moreover, because the time $t^*$ is $S$-good, the set
  $S\in \cS$ must satisfy $\{u,w\}$, implying $\{u,w\} \in \sep(A_1)$
  (where we use that the connection time of $\{u,w\}$ is at least
  $t^* \geq t'$ and $A_1$, $A_2$, and $A_3$ are the only $S$-relevant
  sets in $\cA^{t'}$).  This shows $\sep(U) = \sep (A_1)$, implying
  that $P_S$ covers $U$.

  The case $A_3 \subseteq U$ is also similar. We claim
  $\sep(U) = \sep (A_2)$. To prove this, consider a demand pair
  $\{u,w\}$ in $\sep(A_2)$, where we may assume
  $w\in A_2 \subseteq V\setminus U$, which by the definition of
  autarkic triples implies $u\in A_1 \cup A_3 \subseteq U$.  Then
  $\{u,w\} \in \sep(U)$.  Conversely, consider a demand pair $\{u,w\}$
  in $\sep(U)$.  Because $\cS$ is agreeable and $S\in \cS$ satisfies
  the demand pair $\varphi(U)$, the collection $\cS$ must also satisfy
  $\{u,w\}$.  Moreover, because the time $t^*$ is $S$-good, the set
  $S\in \cS$ must satisfy $\{u,w\}$, implying $\{u,w\} \in \sep(A_2)$
  (where we use that the connection time of $\{u,w\}$ is at least
  $t^* \geq t'$ and $A_1$, $A_2$, and $A_3$ are the only $S$-relevant
  sets in $\cA^{t'}$).  This shows $\sep(U) = \sep (A_2)$, implying
  that $P_S$ covers $U$. 
\end{proof}

\end{document}